\documentclass[12pt]{article}

\usepackage{fullpage}
\usepackage[section]{definitions}
\usepackage[dvips]{color}
\usepackage{pstricks}
\usepackage{pst-node}
\usepackage{algorithm}
\usepackage[noend]{algorithmic}
\usepackage{graphicx}
\usepackage{amsmath}

\newcommand{\Elts}{\ensuremath{E}\xspace}
\newcommand{\Feasible}[1][]{\ensuremath{%
\ifthenelse{\equal{#1}{}}{{\cal F}}{{\cal F}|_{#1}}}\xspace}
\newcommand{\SetSys}{\ensuremath{(\Elts,\Feasible)}\xspace}
\newcommand{\SetSysP}{\ensuremath{(\Elts',\Feasible')}\xspace}
\newcommand{\SetSysPar}[1]{\ensuremath{(\Elts_#1,\Feasible_#1)}\xspace}

\newcommand{\diag}{\ensuremath{\mathop{\mathrm{diag}}}\xspace}
\newcommand{\COST}{\ensuremath{c}\xspace}
\newcommand{\BID}{\ensuremath{b}\xspace}
\newcommand{\NBID}{\ensuremath{x}\xspace}
\newcommand{\PAY}{\ensuremath{p}\xspace}
\newcommand{\vct}[1]{\ensuremath{\mathbf{#1}}\xspace}
\newcommand{\Cost}[2][]{\ensuremath{%
\ifthenelse{\equal{#1}{}}{\COST_{#2}}{\COST^{(#1)}_{#2}}}\xspace}
\newcommand{\Bid}[2][]{\ensuremath{%
\ifthenelse{\equal{#1}{}}{\BID_{#2}}{\BID^{(#1)}_{#2}}}\xspace}
\newcommand{\BidP}[2][]{\ensuremath{%
\ifthenelse{\equal{#1}{}}{\BID'_{#2}}{\BID'^{(#1)}_{#2}}}\xspace}
\newcommand{\NBid}[2][]{\ensuremath{%
\ifthenelse{\equal{#1}{}}{\NBID_{#2}}{\NBID^{(#1)}_{#2}}}\xspace}
\newcommand{\Pay}[2][]{\ensuremath{%
\ifthenelse{\equal{#1}{}}{\PAY_{#2}}{\PAY^{(#1)}_{#2}}}\xspace}
\newcommand{\CostSum}[2][]{\ensuremath{%
\ifthenelse{\equal{#1}{}}{\COST(#2)}{\COST^{(#1)}(#2)}}\xspace}
\newcommand{\BidSum}[2][]{\ensuremath{%
\ifthenelse{\equal{#1}{}}{\BID(#2)}{\BID^{(#1)}(#2)}}\xspace}
\newcommand{\BidSumP}[2][]{\ensuremath{%
\ifthenelse{\equal{#1}{}}{\BID'(#2)}{\BID'^{(#1)}(#2)}}\xspace}
\newcommand{\NBidSum}[2][]{\ensuremath{%
\ifthenelse{\equal{#1}{}}{\NBID(#2)}{\NBID^{(#1)}(#2)}}\xspace}
\newcommand{\PaySum}[2][]{\ensuremath{%
\ifthenelse{\equal{#1}{}}{\PAY(#2)}{\PAY^{(#1)}(#2)}}\xspace}
\newcommand{\COSTVEC}[1][]{\ensuremath{%
\ifthenelse{\equal{#1}{}}{\vct{\COST}}{\vct{\COST}^{(#1)}}}\xspace}
\newcommand{\BIDVEC}[1][]{\ensuremath{%
\ifthenelse{\equal{#1}{}}{\vct{\BID}}{\vct{\BID}^{(#1)}}}\xspace}
\newcommand{\BIDVECP}[1][]{\ensuremath{%
\ifthenelse{\equal{#1}{}}{\vct{\BID'}}{\vct{\BID'}^{(#1)}}}\xspace}
\newcommand{\NBIDVEC}[1][]{\ensuremath{%
\ifthenelse{\equal{#1}{}}{\vct{\NBID}}{\vct{\NBID}^{(#1)}}}\xspace}
\newcommand{\BIDVECEXCL}[2][]{\ensuremath{%
\ifthenelse{\equal{#1}{}}{\vct{\BID}_{-#2}}{\vct{\BID}^{(#1)}_{-#2}}}\xspace}
\newcommand{\LNASH}{\ensuremath{\nu}\xspace}
\newcommand{\LNash}[2][]{\ensuremath{%
\ifthenelse{\equal{#1}{}}{\LNASH(#2)}{\LNASH_{#1}(#2)}}\xspace}
\newcommand{\SNASH}{\ensuremath{\nu^-}\xspace}
\newcommand{\SNash}[2][]{\ensuremath{%
\ifthenelse{\equal{#1}{}}{\SNASH(#2)}{\SNASH_{#1}(#2)}}\xspace}
\newcommand{\Unit}[1]{\ensuremath{\vct{1}_{#1}}\xspace}
\newcommand{\Tot}[1]{\ensuremath{\nu_{#1}}\xspace}
\newcommand{\ADJ}{\ensuremath{A}\xspace}
\newcommand{\DIAG}{\ensuremath{D}\xspace}
\newcommand{\MAT}{\ensuremath{K}\xspace}

\newcommand{\FR}[1][]{\ensuremath{%
\ifthenelse{\equal{#1}{}}{\alpha}{\alpha(#1)}}\xspace}
\newcommand{\PV}{\ensuremath{\vct{q}}\xspace}
\newcommand{\Pv}[1]{\ensuremath{q_{#1}}\xspace}
\newcommand{\MC}{\ensuremath{\vct{\COST'}}\xspace}
\newcommand{\Mc}[1]{\ensuremath{\COST'_{#1}}\xspace}
\newcommand{\MechPay}[2][]{\ensuremath{%
\ifthenelse{\equal{#1}{}}{P(#2)}{P_{#1}(#2)}}\xspace}
\newcommand{\MechRat}[1]{\ensuremath{\phi_{#1}}\xspace}
\newcommand{\SysRat}[1]{\ensuremath{\Phi_{#1}}\xspace}

\newcommand{\MECH}{\ensuremath{\mathcal{M}}\xspace}
\newcommand{\EVMECH}{\ensuremath{\mathcal{EV}}\xspace}
\newcommand{\CONCMECH}{\ensuremath{\mathcal{RCM}}\xspace}
\newcommand{\FLOWMECH}{\ensuremath{\mathcal{FM}}\xspace}
\newcommand{\CUTMECH}{\ensuremath{\mathcal{CM}}\xspace}

\newcommand{\SELRULE}{\ensuremath{\sigma}\xspace}
\newcommand{\SelRule}[1]{\ensuremath{\SELRULE(#1)}\xspace}
\newcommand{\BLOWUP}{\ensuremath{\kappa}\xspace}

\newcommand{\Weight}[1]{\ensuremath{w_{#1}}\xspace}
\newcommand{\VCG}[1][]{\ensuremath{%
\ifthenelse{\equal{#1}{}}{R}{R[#1]}}\xspace}
\newcommand{\MING}{\ensuremath{H}\xspace}

\newcommand{\Nodes}[1]{\ensuremath{N(#1)}\xspace}

\newcommand{\MPC}[2][1]{\ensuremath{%
\ifthenelse{\equal{#1}{}}{\pi(#2)}{\pi_{#1}(#2)}}\xspace}
\newcommand{\TG}{\ensuremath{G'}\xspace}
\newcommand{\MG}{\ensuremath{\tilde{G}}\xspace}
\newcommand{\CYC}{\ensuremath{C}\xspace}
\newcommand{\Cyc}[1]{\ensuremath{\CYC_{#1}}\xspace}
\newcommand{\FLOW}{\ensuremath{F}\xspace}
\newcommand{\FLOWCOLL}{\ensuremath{\mathcal{E}}\xspace}
\newcommand{\CYSET}{\ensuremath{\mathcal{C}}\xspace}
\newcommand{\CYSETF}{\ensuremath{\mathcal{C}^{\rightarrow}}\xspace}
\newcommand{\CYSETB}{\ensuremath{\mathcal{C}^{\leftarrow}}\xspace}
\newcommand{\CYImage}[1]{\ensuremath{\phi(#1)}\xspace}
\newcommand{\EIm}[1]{\ensuremath{\gamma(#1)}\xspace}
\newcommand{\CYSETCOLL}{\ensuremath{\Gamma}\xspace}

\newcommand{\Length}[1]{\ensuremath{\ell_{#1}}\xspace}
\newcommand{\PLength}[1]{\ensuremath{\ell(#1)}\xspace}

\newcommand{\Dist}[1]{\ensuremath{d_{#1}}\xspace}

\newcommand{\dreaches}{\ensuremath{\rightarrow_0}\xspace}

\newcommand{\ALLPATH}{\ensuremath{\mathcal{P}}\xspace}
\newcommand{\ALLCUTS}{\ensuremath{\mathcal{C}}\xspace}
\newcommand{\MINCG}{H\xspace}

\begin{document}

\title{Frugal and Truthful Auctions for Vertex Covers, Flows, and
  Cuts\thanks{A preliminary version of this article appeared in the
    Proceedings of FOCS 2010 \cite{FlowMechanismsFOCS}.}}

\author{David Kempe\\
Department of Computer Science,\\ University of Southern California, CA 90089-0781, USA\\
\texttt{dkempe@usc.edu}
\and
Mahyar Salek\\
Department of Computer Science,\\ University of Southern California, CA 90089-0781, USA\\
\texttt{salek@usc.edu}
\and
Cristopher Moore\\
Computer Science Department and Department of Physics and Astronomy,\\
University of New Mexico, Albuquerque, NM 87131-0001, USA \\
and Santa Fe Institute, Santa Fe NM 87501, USA\\
\texttt{moore@cs.unm.edu} 
}

\maketitle

\begin{abstract}
We study truthful mechanisms for hiring a team of agents in three
classes of set systems: Vertex Cover auctions, $k$-flow auctions, and
cut auctions.
For Vertex Cover auctions, the vertices are owned by selfish and
rational agents, and the auctioneer wants to purchase a vertex cover
from them.
For $k$-flow auctions, the edges are owned by the agents, and
the auctioneer wants to purchase $k$ edge-disjoint $s$-$t$ paths, for
given $s$ and $t$.
In the same setting, for cut auctions, the auctioneer wants to
purchase an $s$-$t$ cut.
Only the agents know their costs, and the auctioneer needs to select a
feasible set and payments based on bids made by the agents.

We present constant-competitive truthful mechanisms for all three set systems.
That is, the maximum overpayment of the mechanism is within a constant
factor of the maximum overpayment of any truthful mechanism, for
\emph{every} set system in the class.
The mechanism for Vertex Cover is based on scaling each bid
by a multiplier derived from the dominant eigenvector of a certain
matrix. The mechanism for $k$-flows prunes the graph to be minimally
$(k+1)$-connected, and then applies the Vertex Cover mechanism.
Similarly, the mechanism for cuts contracts the graph until all $s$-$t$
paths have length exactly 2, and then applies the Vertex Cover
mechanism.

\end{abstract}

\section{Introduction}
Many tasks require the joint allocation of multiple resources
belonging to different bidders. For instance, consider the task of
routing a packet through a network whose edges are owned by different
agents. In this setting, it is necessary to obtain usage rights for
multiple edges simultaneously from the agents.
Similarly, if the agents own the vertices of a graph, and we want to
monitor all edges, we need the right to install monitoring devices on
nodes, and again obtain these rights from distinct agents.

Providing access to edges or nodes in such settings makes the
agents incur a cost \Cost{e}, which the agents should be paid for.
A convenient way to determine ``appropriate'' prices to pay the agents
is by way of \todef{auctions}, wherein the agents $e$ submit bids
\Bid{e} to an \todef{auctioneer}, who selects a \todef{feasible
subset} $S$ of agents to use, and determines prices \Pay{e} to pay the agents.
The most basic case is a single-item auction. The auctioneer requires
the service of any one of the agents, and their services are
interchangeable. Single-item auctions have a long history of study,
and are fairly well understood \cite{klemperer:guide,krishna:auction-theory}.
Motivated by applications in computer networks and electronic
commerce, several recent papers have considered the extension to
a setup termed \todef{hiring a team of agents} \cite{archer:tardos:path-mechanisms,elkind:goldberg:goldberg:frugality,elkind:sahai:steiglitz,BeyondVCG,talwar:price-of-truth}.
In this setting, there is a collection of \todef{feasible sets}, each
consisting of one or more agent.
The auctioneer, based on the agents' bids \Bid{e}, selects one feasible set
$S$, and pays each agent $e \in S$ a price \Pay{e}.

Some of the well-studied special cases of set systems are \todef{path
  auctions}
\cite{archer:tardos:path-mechanisms,elkind:sahai:steiglitz,BeyondVCG,nisan:ronen:algorithmic,yan:truthfulness},
in which the feasible sets are paths from a given source
$s$ to a given sink $t$, and \todef{spanning tree auctions}
\cite{bikchandani:devries:schummer:vohra,garg:kumar:rudra:verma,BeyondVCG,talwar:price-of-truth},
in which the feasible sets are spanning trees of a connected graph.
In both cases, the agents are the edges of the graph.
In this paper, we extend the study to more complex examples of
set systems, namely:
\begin{enumerate}
\item Vertex Covers: The agents are the \emph{vertices}
  of the graph $G$, and the auctioneer needs to select a vertex
  cover \cite{calinescu:truthful,elkind:goldberg:goldberg:frugality,talwar:price-of-truth}.
  Not only are vertex covers of interest in their own right, but they
  give a key primitive for many other set systems as well, an approach
  we explore in depth in this paper.
\item Flows: The agents are the edges of $G$, and the auctioneer
  wants to select $k$ edge-disjoint paths from $s$ to $t$. Thus, this
  scenario generalizes path auctions; the generalization turns
  out to require significant new techniques in the design and analysis
  of mechanisms.
\item Cuts: In the same setting as for flows, the auctioneer wants to
  purchase an $s$-$t$ cut.
\end{enumerate}

In choosing an auction mechanism for a set system, the auctioneer
needs to take into account that the agents are selfish.
Ideally, the auctioneer would like to know
the agents' true costs \Cost{e}. However, the costs are private information,
and a rational and selfish agent will submit a bid
$\Bid{e} \neq \Cost{e}$ if doing so leads to a higher profit.
The area of \todef{mechanism design}
\cite{mas-collel:whinston:green,nisan:ronen:algorithmic,papadimitriou:games}
studies the design of auctions for selfish and rational agents.

We are interested in designing \todef{truthful} (or
\todef{incentive-compatible}) auction mechanisms: auctions under which it is
always optimal for selfish agents to reveal their private costs \Cost{e}
to the auctioneer.
Such mechanisms are societally desirable, because they make the
computation of strategies a trivial task for the agents, and obviate
the need for gathering information about the costs or strategies of
competitors.
They are also desirable from the point of view of analysis, as they
allow us to identify  bids with costs, and let us dispense
with any kinds of assumptions about the distribution of agents' costs.  
Thus, the outcomes of truthful mechanisms are stable in a stronger
sense than Nash equilibria, and may
give bidders more confidence that the right outcome will be reached.
For this reason, truthful mechanism design has been a mainstay
of game theory for a long time.

It is well known that any truthful mechanism will have to pay agents
more than their costs at times; in this paper, we study mechanisms
approximately minimizing the ``overpayment.''
The ratio between the payments of the ``best'' truthful mechanism and
natural lower bounds has been termed the ``Price of Truth'' by Talwar
\cite{talwar:price-of-truth}, and studied in a number of recent papers
\cite{archer:tardos:path-mechanisms,bikchandani:devries:schummer:vohra,elkind:goldberg:goldberg:frugality,elkind:sahai:steiglitz,garg:kumar:rudra:verma,BeyondVCG,talwar:price-of-truth,yan:truthfulness}.
In particular, \cite{BeyondVCG} and \cite{elkind:goldberg:goldberg:frugality}
define and analyze different natural measures of lower bounds on
payments, and define the notions of frugality ratio and
competitiveness.
The \todef{frugality ratio} of a mechanism is the worst-case ratio of
payments to a natural lower bound (formally defined in Section
\ref{sec:preliminaries}), maximized over all cost vectors of the agents.
A mechanism is \todef{competitive} for a class of set systems if its
frugality ratio is within a constant factor of the frugality ratio of
the best truthful mechanism, for \emph{all} set systems in the class.

\subsection{Our Contributions}
In this paper, we present novel frugal mechanisms for three general
classes of set systems: Vertex Covers, $k$-Flows, and Cuts.
Vertex Cover auctions can be considered a very natural primitive for
more complicated set systems. Under the natural assumption that there
are no isolated vertices, they capture set systems with ``minimal
competition'': if the auction mechanism decides to exclude an agent
$v$ from the selected set, this immediately forces the mechanism to
include all of $v$'s neighbors, thus giving these neighbors a
monopoly. Thus, a different interpretation of
Vertex Cover auctions is that they capture any set system whose
feasible sets can be characterized by positive 2SAT formulas: each
edge $(i,j)$ corresponds to a clause $(x_i \vee x_j)$, stating that
any feasible set must include at least one of agents $i$ and $j$.

Our mechanism for Vertex Cover works as follows: based solely on the
structure of the graph $G$, we define an appropriate matrix \MAT and
compute its dominant eigenvector \PV. After agents submit their bids
\Bid{v}, the mechanism first scales each bid to
$\Mc{v} = \Bid{v}/\Pv{v}$, and then simply runs the VCG mechanism
\cite{vickrey:counterspeculation,clarke:multipart,groves:incentives}
with these modified bids. We prove that this mechanism has a frugality
ratio equal to the largest eigenvalue \FR of \MAT,
and that this is within a factor of 2 of the frugality ratio of any mechanism.
The lower bound is based on pairwise competition
between adjacent bidders for \emph{any} truthful mechanism, and in a sense
can be considered the natural culmination 
of the lower bound techniques of \cite{elkind:sahai:steiglitz,BeyondVCG}.
The upper bound is based on carefully balancing all possible worst
cases of a single non-zero cost against each other, and showing that
the worst case is indeed one of these cost vectors.
We stress here that the mechanism does not in general run in
polynomial time: the entries of \MAT are derived from fractional
clique sizes in $G$, which are known to be hard to compute, even
approximately.
We discuss the issue of polynomial time briefly in Section
\ref{sec:conclusions}.

Based on our Vertex Cover mechanism, we present a general methodology
for designing frugal truthful mechanisms. The idea is to take the
original set system, and prune agents from it until it has ``minimal
competition'' in the above sense; subsequently, the Vertex Cover
auction can be invoked. So long as the pruning is ``composable''
in the sense of \cite{aggarwal:hartline:knapsack} (see Section \ref{sec:vertex-cover}),
the resulting auction is truthful. The crux is then to prove that the
pruning step (which removes a significant amount of competition) does
not increase the lower bound on payments too much.
We illustrate the power of this approach with two examples.

\begin{enumerate}
\item For the $k$-flow problem, we show that pruning the graph to a
  minimum-cost $(k+1)$ $s$-$t$-connected graph \MING is composable,
  and increases the lower bound at most by a factor of $k+1$. Hence, we
obtain a $2(k+1)$-competitive mechanism. Establishing the bound of
$k+1$ requires significant technical effort.

\item For the cut problem, we show that pruning the graph to a
  minimum-cost set of edges such that each $s$-$t$ path is cut at
  least twice gives a composable selection rule. Furthermore,
  it increases the lower bound by at most a factor of 2, leading to a
  4-competitive mechanism. For the pruning step, we develop a
  primal-dual algorithm generalizing the Ford-Fulkerson Minimum-Cut
  algorithm.
\end{enumerate}

We note that while the Vertex Cover mechanism is in general not
polynomial, for both special cases derived here, the running time 
is in fact polynomial.

\subsection{Relationship to Past and Parallel Work}
As discussed above, a line of recent papers
\cite{archer:tardos:path-mechanisms,bikchandani:devries:schummer:vohra,elkind:goldberg:goldberg:frugality,elkind:sahai:steiglitz,garg:kumar:rudra:verma,BeyondVCG,talwar:price-of-truth,yan:truthfulness}
analyze frugality of auctions in the ``hiring a team'' setting, where
the auctioneer wants to obtain a feasible set of agents, while paying
not much more than necessary.
In this context, the papers by Karlin, Kempe, and Tamir~\cite{BeyondVCG}
and Elkind, Goldberg, and Goldberg~\cite{elkind:goldberg:goldberg:frugality}
are particularly related to our work.

Karlin et al.~\cite{BeyondVCG} introduce the definitions of frugality
and competitiveness which we use here. They also
give competitive mechanisms for path auctions, and for so-called
$r$-out-of-$k$ systems, in which the auctioneer can select any $r$ out
of $k$ disjoint sets of agents. At the heart of both mechanisms is a
mechanism for $r$-out-of-$(r+1)$ systems. Our mechanism for Vertex
Covers can be considered a natural generalization of this mechanism.
Furthermore, both $r$-out-of-$k$ systems and path auctions are special
cases of $r$-flows, since choosing an $r$-flow in a graph
consisting of $k$ vertex-disjoint $s$-$t$ paths is equivalent to an 
$r$-out-of-$k$ system.
Our approach of pruning the graph is similar in spirit to the approach
in \cite{BeyondVCG}, where graphs were also first pruned to be
minimally 2-connected, and set systems were reduced to
$r$-out-of-$(r+1)$ systems. However, the combinatorial structure of
$k$-flows makes this pruning (and its analysis) much more involved in our case.

Elkind et al.~\cite{elkind:goldberg:goldberg:frugality}
study truthful mechanisms for Vertex Cover. They
present a polynomial-time mechanism with frugality ratio bounded by
$2\Delta$, where $\Delta$ is the maximum degree of the graph,
and also show that there
exist graphs where the best truthful mechanism must have frugality
ratio at least $\Delta/2$. Notice, however, that this does not
guarantee that the mechanism is competitive. Indeed, there are graphs
where the best truthful mechanism has frugality ratio significantly
smaller than $\Delta/2$, and our goal is to have a mechanism which
is within a constant factor of best possible \emph{for every graph}.

Several recent papers have extended the problem of hiring a team of
agents in various directions.  
Cary, Flaxman, Hartline, and Karlin~\cite{cary:flaxman:hartline:karlin}
combine truthful auctions for hiring a team with revenue-maximizing auctions for selling items. 
Du, Sami, and Shi~\cite{du:sami:shi} and 
Iwasaki, Kempe, Saito, Salek, and Yokoo~\cite{FalseNameProof}
study path auctions under the additional requirement that not only
should they be truthful, but \todef{false-name proof}: agents owning
multiple edges have no incentive to claim that these edges belong to
different agents. Du et al.~show that there are no false-name proof
mechanisms that are also Pareto-efficient, and Iwasaki et al.~analyze
the frugality ratio of false-name proof mechanisms, showing
exponential lower bounds.

Results very similar to ours have been derived independently
and simultaneously by 
Chen, Elkind, Gravin, and Petrov~\cite{chen:elkind:gravin:petrov}.
Both papers first derive mechanisms for Vertex Cover auctions.
Our mechanism is based on scaling the agents' bids by 
the entries of the dominant eigenvector of a scaled adjacency matrix.
It has constant competitive ratio for all graphs,
but may not run in polynomial time. The mechanism of Chen et al., on the other
hand, uses eigenvectors of the unscaled adjacency matrix.
It may not be constant competitive on some inputs, but it always runs
in polynomial time. 

Chen et al.~also propose the approach of reducing other set systems to
Vertex Cover instances, called ``Pruning-Lifting Mechanisms'' there.
In particular, they derive the same mechanism as the present paper for
$k$-flows, with similar key lemmas in the proof.
While their Vertex Cover mechanism is different from ours in general,
on inputs derived from flow and cut problems, the scaling factor in
our matrix is the same for all entries, and the mechanisms therefore
coincide. In particular, the mechanisms in both papers are thus
competitive and run in polynomial time.

While the mechanism of Chen et al.~\cite{chen:elkind:gravin:petrov} 
may not always be competitive due to the lack of scaling factors in
the matrix, their proof of a lower bound involves a clever application
of Young's Inequality, and thus avoids losing the factor of 2 in our
lower bound. Thus, whenever their mechanism coincides with ours, both
mechanisms are optimal.
In particular, this also implies that the $k$-flow mechanism of the present
paper is $(k+1)$-competitive and our mechanism for $s$-$t$ cuts is
2-competitive.  
Moreover, they prove stronger
bounds on the $k$-flow mechanism: when compared against the lower
bound from \cite{elkind:goldberg:goldberg:frugality} (used in this
paper, and defined formally in Section \ref{sec:preliminaries}),
the mechanism is in fact optimal.

Finally, in collaboration with the authors of
\cite{chen:elkind:gravin:petrov}, we recently showed that 
Young's Inequality can be applied to the analysis of our Vertex
Cover mechanism, removing the factor of 2 from the lower bound. 
In other words, we show that our Vertex Cover mechanism is indeed
optimal for all Vertex Cover instances. 
This result will be included in a joint full version of both papers.

\section{Preliminaries} \label{sec:preliminaries}
A set system \SetSys has $n$ \todef{agents} (or \todef{elements}),
and a collection $\Feasible \subseteq 2^{\Elts}$ of \todef{feasible sets}.
We call a set system \todef{monopoly-free} if no element is in all
feasible sets, i.e., if $\bigcap_{S \in \Feasible} S = \emptyset$.
The three classes of set systems studied in this paper are:
\begin{enumerate}
\item Vertex Covers: here, the agents are the \emph{vertices} of a
  graph $G$, and \Feasible is the collection of all vertex covers of
  $G$. To avoid confusion, we will denote the agents by $u,v$ instead
  of $e$ in this case. Notice that every Vertex Cover set system is
  monopoly-free.
\item $k$-flows: here, we are given a graph $G$ with source $s$ and
  sink $t$. The agents are the \emph{edges} of $G$. A set of edges is
  feasible if it contains at least $k$ edge-disjoint $s$-$t$ paths.
  A $k$-flow set system is monopoly-free if and only if the minimum
  $s$-$t$ cut cuts at least $k+1$ edges.
\item Cuts: With the same setup as for $k$-flows, a set of edges is
  feasible if it contains an $s$-$t$ cut. Thus, the set system is
  monopoly-free if and only if $G$ contains no edge from $s$ to $t$.
\end{enumerate}

The set system \SetSys is common knowledge to the auctioneer and all agents.
Each agent $e \in E$ has a \todef{cost} $\Cost{e}$, which is private, i.e.,
known only to $e$.  
We write $\CostSum{S} = \sum_{e \in S} \Cost{e}$ for the
total cost of a set $S$ of agents, and also extend this notation to
other quantities (such as bids or payments).
A \todef{mechanism} for a set system proceeds as follows:

\begin{enumerate}
\item Each agent submits a sealed bid \Bid{e}.
\item Based on the bids \Bid{e}, the auctioneer selects a feasible set
$S \in \Feasible$ as the winner, and computes a payment $\Pay{e} \geq
\Bid{e}$ for each agent $e \in S$.
The agents $e \in S$ are said to \todef{win}, while all other agents \todef{lose}.
\end{enumerate}

Each agent, knowing the algorithm for computing the winning set and
the payments, will choose a bid \Bid{e} maximizing her own \todef{profit}, 
which is $\Pay{e} - \Cost{e}$ if the agent wins, and 0 otherwise.
We are interested in mechanisms where self-interested agents will bid
$\Bid{e} = \Cost{e}$. More precisely, a mechanism is \todef{truthful}
if, for any fixed vector \BIDVECEXCL{e} of bids by all other agents,
$e$ maximizes her profit by bidding $\Bid{e} = \Cost{e}$.
If a mechanism is known to be truthful, we can use $\Bid{e}$ and
$\Cost{e}$ interchangeably.
It is well-known \cite{archer:tardos:path-mechanisms,krishna:auction-theory}
that a mechanism is truthful only if its selection rule is
\todef{monotone} in the following sense: if all other agents' bids
stay the same, then a losing agent cannot become a winner by raising
her bid. Once the selection rule is fixed, there is a unique payment
scheme to make the mechanism truthful. Namely, each agent is
paid her \todef{threshold bid}: the supremum of all winning bids she
could have made given the bids of all other agents.


\subsection{Nash Equilibria and Frugality Ratios}
\label{sec:Nash}

To measure how much a truthful mechanism ``overpays,'' we need a
natural bound to compare the payments to.
Karlin et al.~\cite{BeyondVCG} proposed using as a bound the
solution of a natural minimization problem.
Let $S$ be the cheapest feasible set with respect to the
true costs \Cost{e}; ties are broken lexicographically.

\begin{LP}[eqn:cheapest-nash]{Minimize}{\SNash{\COSTVEC} := \sum_{e\in S} \NBid{e}}
\NBid{e} \geq \Cost{e} & \mbox{ for all } e \in S \\
\NBid{e} = \Cost{e} & \mbox{ for all } e \notin S\\
\sum_{e \in S} \NBid{e} \leq \sum_{e \in T} \NBid{e} & 
\mbox{ for all } T \in \Feasible\\[0.5ex]
\multicolumn{2}{l}{\mbox{For every $e\in S$, there is a $T_e \in
  \Feasible, e \notin T_e$ such that}}\\
\sum_{e' \in S} \NBid{e'} = \sum_{e'\in T_e} \NBid{e'}
\end{LP}

The intuition for this optimization problem is that it captures the
bids of agents in the cheapest ``Nash Equilibrium'' of a first-price
auction with full information, under the assumption that the actual
cheapest set $S$ wins, and the losing agents all bid their costs.
That is, the mechanism selects the cheapest set with respect to the
bids \NBid{e}, and pays each winning agent her bid \NBid{e}.
The first constraint captures individual rationality. The third
constraint states that the bids \NBid{e} are such that $S$
still wins, and the final
constraint states that for each winning agent, there is a
\todef{tight set} preventing her from bidding higher.  
That is, if $e$ increases her bid, the buyer will select a set $T$ excluding $e$ instead of $S$.
We say that a vector \NBIDVEC is \todef{feasible} if it satisfies all
these constraints. 

While this optimization problem is inspired by the 
analogy of Nash Equilibria, it should be noted that first-price
auctions do not in general have Nash Equilibria due to tie-breaking
issues (see a more detailed discussion in
\cite{immorlica:karger:nikolova:sami,BeyondVCG}).  

Elkind et al.~\cite{elkind:goldberg:goldberg:frugality} and Chen and
Karlin \cite{chen:karlin:cheap-labor} observed that the quantity 
\SNash{\COSTVEC} has several undesirable non-monotonicity properties. 
For instance, 
adding new feasible sets to the set system, and thus increasing 
the amount of competition between agents,
can sometimes lead to higher values of \SNash{\COSTVEC}. 
Similarly, lowering the costs of losing agents, or increasing the
costs of winning agents, can sometimes increase \SNash{\COSTVEC}.
Furthermore, \SNash{\COSTVEC} is NP-hard to compute even if the set
system is the set of all $s$-$t$ paths \cite{chen:karlin:cheap-labor}.

Instead, Elkind et al.~\cite{elkind:goldberg:goldberg:frugality}
propose replacing the minimization by a maximization in the above
optimization problem. 
An important advantage of this optimization problem is that the
maximization objective ensures that for every $e \in S$, there is a
tight set $T$. Thus, the maximization objective removes the need for
the final constraint, and turns the optimization problem into an 
instance of Linear Programming, which can
be solved in many cases. We thus obtain the following definition
(which \cite{elkind:goldberg:goldberg:frugality} refers to as 
$\mathrm{NTU}_{\max}$). 
Intuitively, this definition captures the bids in the most
\emph{expensive} Nash Equilibrium of a first-price auction, with the
same caveat as before about the non-existence of equilibria.

\begin{LP}[eqn:nash-def]{Maximize}{\LNash{\COSTVEC} := \sum_{e\in S} \NBid{e}}
(\mbox{i})\;\; \NBid{e} \geq \Cost{e} & \mbox{ for all } e \\
(\mbox{ii})\; \NBid{e} = \Cost{e} & \mbox{ for all } e \notin S \\
(\mbox{iii}) \sum_{e \in S} \NBid{e} \leq \sum_{e \in T} \NBid{e} &
\mbox{ for all } T \in \Feasible
\end{LP}

\noindent
As stated above, a consequence of this maximization is that, for every
$e$ in the winning set, there is a tight set $T$ excluding $e$ that
prevents $e$ from bidding higher: 
\begin{equation}
\label{eq:tight}
\forall e \in S : \exists T \in \Feasible: e \notin T \text{ and } \sum_{e' \in S} \NBid{e'} = \sum_{e' \in T} \NBid{e'} \, . 
\end{equation}

We will refer to the bounds $\SNash{\COSTVEC}$ and $\LNash{\COSTVEC}$
as \emph{buyer-optimal} and \emph{buyer-pessimal}, respectively, throughout the
paper. Moreover, due to the advantages discussed above, we will use
the quantity \LNash{\COSTVEC} as a natural lower bound for this paper.  
Despite the preceding discussion, in order to emphasize the intuition
behind the bounds, we will refer to the \NBid{e} values of the LP~\eqref{eqn:nash-def}
as the \todef{Nash Equilibrium bids} of agents $e$,
or simply the \todef{Nash bids} of $e$.

Notice that \LNash{\COSTVEC} is defined
for all monopoly-free set systems.
We now formally define the frugality ratio of a mechanism \MECH for a
set system \SetSys, and the notion of a competitive mechanism.

\begin{definition}[Frugality Ratio, Competitive Mechanism]
Let \MECH be a truthful mechanism for the set
system \SetSys, and let
\MechPay[\MECH]{\COSTVEC} denote the total payments of
\MECH when the vector of actual costs is \COSTVEC.

\begin{enumerate}
\item The \todef{frugality ratio} of \MECH is
\begin{eqnarray*}
\MechRat{\MECH}
& = & \sup_{\COSTVEC} \frac{\MechPay[\MECH]{\COSTVEC}}{\LNash{\COSTVEC}}.
\end{eqnarray*}

\item The frugality ratio of the set system \SetSys is
\begin{eqnarray*}
\SysRat{\SetSys}
& = & \inf_{\MECH} \MechRat{\MECH},
\end{eqnarray*}
where the infimum is taken over all truthful mechanisms \MECH for
\SetSys.

\item A mechanism \MECH is \todef{\BLOWUP-competitive} for a class of set
  systems \SET{\SetSysPar{1}, \SetSysPar{2}, \ldots} if
  \MechRat{\MECH} is within a factor \BLOWUP of
  \SysRat{\SetSysPar{i}} for all $i$.
\end{enumerate}
\end{definition}

\begin{remark}
The frugality ratio of a mechanism is defined as instance-based.
The frugality ratio of a set system captures the inherent structural
complexity of that instance, which can be ``exploited'' with careful
worst-case choices of costs.

Competitiveness, on the other hand, is defined over a class of set
systems. If a single mechanism, such as the ones defined in this
paper, is competitive, it does as well on each set system in the
class as the best mechanism, which could possibly be tailored to this
specific instance. The nomenclature ``competitive'' is motivated by
the analogy with online algorithms.

The instance-based definition \cite{BeyondVCG,elkind:goldberg:goldberg:frugality}
allows us a more fine-grained distinction between mechanisms than
earlier work (e.g., \cite{archer:tardos:path-mechanisms,nisan:ronen:algorithmic}),
where a lower bound in terms of a worst case over all instances was used.
\end{remark}

As discussed above, the motivation for the LPs 
\eqref{eqn:cheapest-nash} and \eqref{eqn:nash-def}
was that they provide ``natural lower bounds'' on the payments of any
truthful mechanism. However, to the best of our knowledge, it was
previously unknown whether the solutions do in fact provide lower bounds. 
Indeed, it is easy to define mechanisms that achieve arbitrarily lower
payments for particular cost vectors, albeit at the cost of
significantly higher payments on other cost vectors. Here, we
establish that the objective value of the LP \eqref{eqn:nash-def}
indeed does give a lower bound in terms of the frugality ratio.
This resolves an open question from the preliminary version 
of this paper \cite{FlowMechanismsFOCS}.

\begin{proposition}
Let \SetSys be an arbitrary set system and \MECH a truthful and
individually rational mechanism on \SetSys.
Then, $\MechRat{\MECH} \geq 1$.
\end{proposition}  

\begin{proof}
Let \COSTVEC be an arbitrary cost vector.
Let $S \in \Feasible$ be the set minimizing \CostSum{S}, 
and \NBIDVEC the solution to the LP \eqref{eqn:nash-def}.
Let $S' \in \Feasible$ be the winning set for \MECH with cost vector \NBIDVEC. 
Because \MECH is truthful and individually rational, its payment 
$\MechPay[\MECH]{\NBIDVEC}$ is at least $\sum_{e \in S'} \NBid{e}$.
By the third constraint of the LP \eqref{eqn:nash-def},
$\sum_{e \in S'} \NBid{e} \geq \sum_{e \in S} \NBid{e}$.
Finally, by construction, we have that $\LNash{\COSTVEC} = \LNash{\NBIDVEC}$. 
Taken together, this implies that
\[ 
\MechPay[\MECH]{\NBIDVEC} \; \geq \; \sum_{e \in S'} \NBid{e} 
\; \geq \; \sum_{e \in S} \NBid{e} 
\; = \; \LNash{\NBIDVEC}.
\]
By definition of the frugality ratio, this implies that 
$\MechRat{\MECH} \geq 1$.
\end{proof}

\section{Vertex Cover Auctions} \label{sec:vertex-cover}

In this section, we describe and analyze a constant-competitive
mechanism for Vertex Cover auctions. We then show how to use it as the
basis for a methodology for designing frugal mechanisms for other set
systems. The graph is denoted by $G=(V,E)$,
with $n$ vertices. We write $u \sim v$ to denote that $(u,v) \in E$.

Our mechanmism is based on certain modifications to the
well-known Vickrey-Clarke-Groves (VCG) mechanism
\cite{vickrey:counterspeculation,clarke:multipart,groves:incentives}.
Recall that VCG always selects the cheapest feasible set $S$ with
respect to the submitted bids \Bid{e}, and pays each agent her
threshold bid.

\subsection{Weighting the bids with an eigenvector}

The important change to VCG in our mechanism is that each
  agent's bid is scaled by an agent-specific multiplier.
The multipliers capture ``how important'' an agent is for the solution, 
roughly in the sense of how many other agents can be omitted
by including this agent. They are computed as entries of the
  dominant eigenvector of a certain matrix \MAT.
As we will see, the computation of \MAT is NP-hard itself, so
the mechanism will in general not run in polynomial time unless P=NP.

As a first step, our mechanism removes all isolated vertices.
We assume that the resulting graph $G$ is connected.
Let \Unit{v} (for any vertex $v$) be the vector with 1 in coordinate
$v$ and 0 in all other coordinates.
We define $\Tot{v} = \LNash{\Unit{v}} \geq 1$ to be
the total ``Nash Equilibrium'' payment of the first-price auction
in the sense of the LP~\eqref{eqn:nash-def} 
if agent $v$ has cost 1 and all other agents have cost 0.
Notice that in this case, $v$ loses.
We prove in Section 3.1 that \Tot{v} is exactly the fractional clique
number of the graph induced by the neighbors of v, without v itself.
This implies
that unless ZPP=NP,
\Tot{v} cannot be approximated to within a factor
$O(n^{1-\epsilon})$ in polynomial time, for any $\epsilon > 0$.
Our inability to compute \Tot{v} is the chief obstacle to a
constant-competitive polynomial-time mechanism.

Let \ADJ be the adjacency matrix of $G$ (with diagonal 0).
Define $\DIAG = \diag (1/\Tot{1}, 1/\Tot{2}, \ldots, 1/\Tot{n})$, and
$\MAT = \DIAG \ADJ$.  That is, 
\begin{eqnarray*}
\MAT_{u,v} & = & 
\begin{cases} 1/\Tot{u} & \mbox{if } u \sim v \\ 0 & \mbox{if } u \not\sim v \, . \end{cases}
\end{eqnarray*}
If we define $\MAT' = \DIAG^{-1/2} \MAT \DIAG^{1/2} = \DIAG^{1/2} \ADJ \DIAG^{1/2}$, 
then $\MAT$ and $\MAT'$ have the same eigenvalues, and the eigenvectors of \MAT are of the form 
$\DIAG^{1/2} \cdot \vct{e}$, where \vct{e} is an eigenvector of $\MAT'$.  
Moreover, 
\begin{eqnarray*}
\MAT'_{u,v} & = & \begin{cases} 1/\sqrt{\Tot{u} \Tot{v}} & \mbox{if } u \sim v 
                             \\ 0 & \mbox{if } u \not\sim v \, , \end{cases}
\end{eqnarray*}
so $\MAT'$ is symmetric and has non-negative entries.
By the Perron-Frobenius Theorem, the eigenvalues of \MAT' and \MAT are real.  
Since we assumed $G$ to be connected, the 
dominant eigenvector of $\MAT'$ is unique and has positive entries,
and the same holds for \MAT.

Let \FR be the largest eigenvalue of \MAT, and 
\PV the corresponding eigenvector.
Notice that given \MAT as input, \FR and \PV can be computed
efficiently and without knowledge of the agents' bids or costs.

The mechanism \EVMECH (which stands for ``Eigenvector
  Mechanism'') is now as follows:
after all nodes $v$ submit their bids \Bid{v},
the algorithm sets $\Mc{v} = \Bid{v}/\Pv{v}$,
and computes a minimum cost vertex cover $S$ with respect to the
costs \Mc{v} (ties broken lexicographically).
$S$ is chosen as the winning set, and each agent in $S$ is paid her
threshold bid. 
Notice that the second step of the mechanism again requires the
solution to an NP-hard problem.

\EVMECH is truthful since the selection rule
is clearly monotone, and the payments are the threshold bids.
Thus, we can assume without loss of generality that bids and costs coincide.
In the following, we analyze the frugality ratio of \EVMECH, and show
that \EVMECH is competitive.

\begin{lemma} \label{lem:vertex-frugality}
\EVMECH has frugality ratio at most \FR.
\end{lemma}

\begin{emptyproof}
We start by considering only cost vectors with exactly one
non-zero entry, i.e., of the form $\COSTVEC = \Cost{v} \cdot \Unit{v}$.
For such a cost vector, the mechanism will choose a subset of
$V \setminus \SET{v}$ as the winning set, and pay each $u$ in that
subset her threshold bid. We calculate the threshold bids of all these
agents $u$.

First, consider any agent $u \sim v$.
If $u$ were to raise her bid above $(\Pv{u} / \Pv{v}) \cdot \Cost{v}$,
while all agents besides $u$ and $v$ continued to bid 0, 
then the set $V \setminus \SET{u}$ would be cheaper than \SET{u}
with respect to the new bid vector \MC.
Therefore, $u$ would not be part of the winning vertex cover. 
Thus $u$'s threshold payment is at most
$(\Pv{u} / \Pv{v}) \cdot \Cost{v}$. 
 
Next, consider any agent $u \not\sim v$.
Because $V \setminus \SET{u,v}$ is a vertex cover, $u$ cannot raise
her bid above zero without losing, so her threshold bid is $0$.
Hence, the total payment of \EVMECH is at most
$\MechPay{\COSTVEC} = (1 / \Pv{v}) \cdot \Cost{v} \cdot \sum_{u \sim v} \Pv{u}$.
On the other hand, by the definition of \Tot{v} and linearity of \LNASH,
we have that $\LNash{\COSTVEC} = \Cost{v} \Tot{v}$, so the
frugality ratio for cost vectors of the form $\Cost{v} \cdot \Unit{v}$ is
\[ 
\frac{(1/\Pv{v}) \cdot \Cost{v} \cdot \sum_{u \sim v} \Pv{u}}{\Cost{v} \Tot{v}}
\; = \;
\frac{1}{\Pv{v}} \cdot \sum_{u \sim v} \frac{1}{\Tot{v}} \cdot \Pv{u}
\; = \; \frac{1}{\Pv{v}} \cdot \FR \cdot \Pv{v}
\; = \; \FR,
\]
where the second equality followed because the vector \PV is an
eigenvector of \MAT with eigenvalue \FR. Thus, for any cost vector
with only one non-zero entry, the frugality ratio is at most \FR.

Now consider an arbitrary cost vector \COSTVEC, and write it as
$\COSTVEC = \sum_v \Cost{v} \Unit{v}$.
We claim that
$\MechPay{\COSTVEC} \leq \sum_v \Cost{v} \MechPay{\Unit{v}}$.
For consider any vertex $u \in S$ winning with cost vector \COSTVEC.  
If the cost vector were $\Cost{v} \Unit{v}$ instead, $u$'s payment 
would be $(\Pv{u}/\Pv{v}) \cdot \Cost{v}$ if $u \sim v$ and 0 otherwise.
On the other hand, when the cost vector is \COSTVEC, if $u$ bids
strictly more than $\sum_{v \sim u} \Pv{u}/\Pv{v} \cdot \Cost{v}$, then
$u$ cannot be in the winning set, as replacing $u$ with all its neighbors
would give a cheaper solution with respect to the costs \MC.
Thus, each node $u$ gets paid at most $\sum_{v \sim u} \Pv{u}/\Pv{v} \cdot \Cost{v}$
with cost vector \COSTVEC, and the total payment is at most
\[ 
\MechPay{\COSTVEC}
\; = \; \sum_u \sum_{v \sim u} \frac{\Pv{u}}{\Pv{v}} \cdot \Cost{v}
\; = \; \sum_v \Cost{v} \cdot \sum_{u \sim v} \frac{\Pv{u}}{\Pv{v}}
\; = \; \sum_v \Cost{v} \MechPay{\Unit{v}}.
\] 

On the other hand, we have that
\[ 
\LNash{\COSTVEC}
\; \geq \; \sum_v \Cost{v} \LNash{\Unit{v}}
\; = \; \sum_v \Cost{v} \Tot{v},
\] 
because of the following argument: for each $v$, let \NBIDVEC[v] be a
an optimal solution for the LP~\eqref{eqn:nash-def} with cost vector
\Unit{v}. Then, simply by linearity, the vector
$\NBIDVEC = \sum_v \Cost{v} \NBIDVEC[v]$ is feasible for~\eqref{eqn:nash-def} with cost vector \COSTVEC, and achieves
the sum of the payments. Thus, the optimal solution to~\eqref{eqn:nash-def} with cost vector \COSTVEC can have no smaller
total payments.

Combining the results of the previous two paragraphs, we have the
following bound on the frugality ratio:
\[ 
\max_{\COSTVEC} \frac{\MechPay{\COSTVEC}}{\LNash{\COSTVEC}}
\; \leq \; \max_{\COSTVEC} \frac{\sum_v \Cost{v}
  \MechPay{\Unit{v}}}{\sum_v \Cost{v} \Tot{v}}
\; \leq \; \max_v \frac{\MechPay{\Unit{v}}}{\Tot{v}}
\; \leq \; \FR. \qquad \qquad \QED
\] 
\end{emptyproof}

Next, we prove that no other mechanism can do asymptotically better.
\begin{lemma} \label{lem:vertex-lower}
Let \MECH be any truthful vertex cover mechanism on $G$.
Then, \MECH has frugality ratio at least $\frac{\FR}{2}$.
\end{lemma}

\begin{proof}
We construct a directed graph $G'=(V,E')$ from $G$ by directing each edge $e$
of $G$ in at least one direction. Consider any edge $e=(u,v)$ of $G$.
Let \COSTVEC be the cost vector in which
$\Cost{u} = \Pv{u}, \Cost{v} = \Pv{v}$, and $\Cost{i} = 0$ for all
$i \neq u,v$.
When \MECH is run on the cost/bid vector \COSTVEC, at least one of $u$
and $v$ must be in the winning set $S$; otherwise, it would not be a
vertex cover.
If $u \in S$, then add the directed edge $(v,u)$ to $E'$. Similarly,
if $v \in S$, then add $(u,v)$ to $E'$. If both $u,v \in S$, then
add both directed edges. By doing this for all edges $e \in G$, we
eventually obtain a graph $G'$.

Now give each node $v$ a weight $\Pv{v}$.
Each node-weighted directed graph $(V,E')$
contains at least one node $v$ such that 
\begin{eqnarray*}
\sum_{u: (v,u) \in E'} \Pv{u} 
& \geq & \sum_{u: (u,v) \in E'} \Pv{u} \, , 
\end{eqnarray*}
(see, e.g., the proof of Lemma 11 in \cite{BeyondVCG}), and hence
\begin{eqnarray*}
\sum_{u: (v,u) \in E'} \Pv{u} 
& \geq & \frac{1}{2} \sum_{u: u \sim v} \Pv{u} \, .
\end{eqnarray*}
Fix any such node $v$ in $G'$ with respect to the weights \Pv{v}.

Now consider the cost vector \COSTVEC with $\Cost{v} = \Pv{v}$ and
$\Cost{i} = 0$ for all $i \neq v$. By monotonicity of the
selection rule of \MECH (which follows from the truthfulness of \MECH),
at least all nodes $u$ such that $(v,u) \in G'$ must be part of the
selected set $S$ of \MECH, and must be paid at least \Pv{u}.
Therefore, the total payment of \MECH is at least
\[ 
\sum_{u: (v,u) \in G'} \Pv{u}
\; \geq \; \half \sum_{u \sim v} \Pv{u}
\; = \; \half \Tot{v} \sum_{u \sim v} \frac{1}{\Tot{v}} \Pv{u}
\; = \; \half \Tot{v} \cdot \FR \Pv{v},
\] 
where the last equality followed from the fact that \PV is an
eigenvector of the matrix \MAT.

On the other hand, as in the proof of Lemma \ref{lem:vertex-frugality},
$\LNash{\COSTVEC} = \Tot{v} \Pv{v}$ for our cost vector \COSTVEC, so
the frugality ratio 
is at least $\half \FR$, when the cost vector is \COSTVEC.
\end{proof}

Combining Lemma \ref{lem:vertex-frugality} and Lemma
\ref{lem:vertex-lower}, we have proved the following theorem:

\begin{theorem}
\EVMECH is 2-competitive for Vertex Cover auctions.
\end{theorem}

\begin{remark}
The lower bound of $\half \FR$ on the frugality ratio of any
  mechanism can potentially be large. For instance, for a complete
  bipartite graph $K_{n,n}$, we have $\FR = \Theta(n)$. Thus, such
  large overpayments are inherent in truthful mechanisms in general.
  However, truthful mechanisms may be much more frugal on specific
  classes of graphs.
\end{remark}
 
\begin{remark}
\EVMECH in general does not run in polynomial time. For the
  final step, computing a minimum-cost vertex cover with respect to
  the scaled costs, we could use a monotone 2-approximation, as
  suggested by Elkind et al.~\cite{elkind:goldberg:goldberg:frugality}.
  The hardness of computing \MAT is more severe. However, notice that
  for specific classes of graphs, such as degree-bounded or
  triangle-free graphs, \MAT can be computed efficiently, giving us
  non-trivial polynomial-time mechanisms for Vertex Cover on those
  classes. This issue is discussed more in Section \ref{sec:conclusions}.
\end{remark}


\subsection{Nash Equilibria and the Fractional
 Clique Problem} \label{sec:fractional-clique}

In this section, we show that the Nash Equilibrium values \Tot{v} used
for scaling of the matrix actually have a natural interpretation.
To state the result, recall that the \todef{fractional clique
number} is the solution to the linear program
\begin{LP}[eqn:fractional-clique]{Maximize}{\sum_{u} \NBid{u}}
\sum_{u \in I} \NBid{u} \leq 1 & \mbox{ for all independent sets } I \\
\NBid{u} \geq 0 & \mbox{ for all } u\\
\end{LP}%
The \todef{fractional chromatic number} is the solution of the dual problem, 
where we have a variable $y_I$ for each independent set $I$ and a constraint 
$\sum_{I \ni u} y_I \geq 1$ for each vertex $u$, and we minimize $\sum_I y_I$.  
By LP duality, the fractional clique number and the fractional chromatic number are equal.

\begin{proposition}
\label{prop:clique}
Let $G_v$ be the subgraph induced by the neighborhood of $v$ but
without $v$ itself. Then, $\Tot{v}$ is exactly the fractional clique
number, and thus the fractional chromatic number, of $G_v$.
\end{proposition}

\begin{proof}
Let \NBIDVEC be any bid vector feasible for the LP~\eqref{eqn:nash-def}. 
First, for all vertices $u$ that do not share an
edge with $v$, we must have $\NBid{u} = 0$, because
$V \setminus \SET{u,v}$ is a feasible set.
So we can restrict our attention to $G_v$.  

For a set $I$, we write $\NBidSum{I} = \sum_{u \in I} \NBid{u}$ for
the total bids of the vertices in $I$.
If $I$ is an independent set in $G_v$, then $\NBidSum{I} \leq 1$.
The reason is that the set $V \setminus I$ is also feasible, and would
cost less than $V \setminus \SET{v}$ if \NBidSum{I} exceeded 1.
Thus, any feasible bid vector \NBIDVEC induces a feasible solution to
the LP~\eqref{eqn:fractional-clique}, of the same total cost.

Conversely, if we have a feasible solution to the LP~\eqref{eqn:fractional-clique}, 
we can extend it to a bid vector for
all agents by setting $\NBid{v} = 1$, and $\NBid{u} = 0$ for all
vertices $u$ outside $v$'s neighborhood.
We need to show that each feasible set $T$, i.e., each vertex cover,
has total bid $\NBidSum{T}$ at least as large as the set $V \setminus \SET{v}$.
If $T$ does not contain $v$, it must contain all of $v$'s neighbors;
it thus has the same bid as $V \setminus \SET{v}$ by definition.
Otherwise, because $V \setminus T$ is an independent set,
the feasibility of \NBIDVEC for
the LP~\eqref{eqn:fractional-clique} implies that
$\NBidSum{V \setminus T} \leq 1$.
Thus, $\NBidSum{T} \geq \NBidSum{V} -1 = \NBidSum{V \setminus \SET{v}}$, 
and the two LPs~\ref{eqn:nash-def} and
\eqref{eqn:fractional-clique} have the same value.
\end{proof}

Standard randomized rounding arguments (see, e.g., \cite{lund:yannakakis:minimization}) 
imply that for any graph, the chromatic number and the fractional
chromatic number are within a factor $O(\log n)$ of each other.
Therefore,
any approximation hardness results for Graph Coloring also apply to the
Fractional Clique Problem with at most a loss of logarithmic
factors. In particular, the result of Feige and
Kilian~\cite{feige:kilian:zero-knowledge} implies that
unless ZPP=NP, \Tot{v} cannot be approximated to within a factor
$O(n^{1-\epsilon})$ in polynomial time, for any $\epsilon > 0$.


\subsection{Composability and a General Design Approach}
\label{sec:strongly-monotone}

Vertex Cover auctions can be used naturally as a way to deal with
other types of set systems.  First, pre-process the set system by
removing a subset of agents, turning the remaining set system into a
Vertex Cover instance; then, run \EVMECH on that instance.

The important part is to choose the pre-processing rule to ensure
that the overall mechanism is both truthful and competitive.
A condition termed \todef{composability} in
\cite[Definition 5.2]{aggarwal:hartline:knapsack} is sufficient to
ensure truthfulness.   We show here that a comparison between lower
bounds is sufficient to show competitiveness.

\begin{definition}[Composability \cite{aggarwal:hartline:knapsack}]
Let \SELRULE be a selection rule mapping bid vectors to 
subsets of (remaining) agents. We say that \SELRULE is \todef{composable}
if $\SelRule{\BIDVEC} = T$ implies that
$\SelRule{\BidP{e}, \BIDVECEXCL{e}} = T$ for
any $e \in T$ and $\BidP{e} \leq \Bid{e}$.
In other words, not only can a winning agent not become a loser by
bidding lower; she cannot even change \emph{which} set containing
her wins.
\end{definition}

Formally, when we talk about ``removing'' a set of agents from a set
system, we are replacing \SetSys with $(T, \Feasible[T])$, where
$T = \SelRule{\BIDVEC}$, and
$\Feasible[T] := \Set{S \in \Feasible}{S \subseteq T}$.

\begin{theorem} \label{thm:reduction}
Let \SELRULE be a composable selection rule with the following
additional property:
For all monopoly-free set systems \SetSys in the class,
and all cost vectors \COSTVEC, writing
$\SetSysP := (\SelRule{\COSTVEC}, \Feasible[\SelRule{\COSTVEC}])$:
\begin{enumerate}
\item \SetSysP is a Vertex Cover instance, and
\item $\LNash[\SetSysP]{\COSTVEC}
\leq \BLOWUP \cdot \LNash[\SetSys]{\COSTVEC}$.
\end{enumerate}

Let the \todef{Remove-Cover Mechanism} \CONCMECH consist of running
\EVMECH on \SetSysP.
Then \CONCMECH is a truthful $2\BLOWUP$-competitive mechanism.
\end{theorem}

\begin{proof}
Truthfulness is proved in \cite[Lemma 5.3]{aggarwal:hartline:knapsack}.
The proof is short, and we include a version here for completeness.
Consider any agent $e$, and a bid vector \BIDVECEXCL{e} for agents
other than $e$. Because \SELRULE is composable, and thus also
monotone, there is a threshold bid $\tau_e$ such that $e$ wins iff her
bid is at most $\tau_e$.  Furthermore, whenever $\Bid{e} \leq \tau_e$,
the set \SelRule{\BIDVEC} is uniquely determined, and independent of
\Bid{e}. Thus, whenever $\Bid{e} \leq \tau_e$, \EVMECH will be run on
the same set system
$(\SelRule{\BIDVEC}, \Feasible[\SelRule{\BIDVEC}])$, and the selection
rule of \EVMECH on this set system is monotone. Hence, the overall
selection rule of \CONCMECH is monotone for $e$, 
implying directly that \CONCMECH is truthful.

The upper bound on the frugality ratio of \CONCMECH follows simply
from Lemma \ref{lem:vertex-frugality} and the assumption of the theorem:
\[ 
\MechPay[\CONCMECH]{\COSTVEC}
\; \leq \; \FR[\SetSysP] \cdot \LNash[\SetSysP]{\COSTVEC}
\; \leq \; \FR[\SetSysP] \cdot \BLOWUP \cdot \LNash[\SetSys]{\COSTVEC}.
\] 

To prove the lower bound, let \MECH be any truthful mechanism for
\SetSys, and let \SetSysP be the Vertex Cover set system maximizing
\FR[\SetSysP].
We consider cost vectors \COSTVEC with $\Cost{e} = \infty$ (or some
very large finite values) for $e \notin \Elts'$.
For such cost vectors, we can safely disregard all elements $e \notin
\Elts'$ altogether, as they will not affect the solutions to the LP~\eqref{eqn:nash-def}, 
nor be part of any solution selected by \MECH.

But then, \MECH is exactly a mechanism selecting a feasible solution
to the Vertex Cover instance \SetSysP.
By Lemma \ref{lem:vertex-lower}, \MECH thus has frugality ratio at
least $\FR[\SetSysP]/2$, completing the proof.
\end{proof}

\Omit{
We note that composability is crucial in the proof of Theorem
\ref{thm:reduction}.
If \SELRULE were merely monotone, then an agent $e$'s bid might
alter the set \SelRule{\BIDVEC} in such a way that the subsequent
threshold bid of $e$ in \EVMECH increases, and $e$ thus gets paid
more. While it may in principle be possible to still prove
monotonicity of \CONCMECH in some special cases, such a proof would
likely be much more involved than our proposed approach.
}

A simple general way to obtain a composable rule is to 
choose the set with the minimum total cost, from some subset of the feasible sets:

\begin{lemma} \label{lem:composable}
Let \SELRULE be any rule with consistent tie breaking
selecting a set $S$ minimizing \BidSum{S} over all sets $S$ with a
certain property $P$. 
Then \SELRULE is composable.
\end{lemma}

\begin{proof}
Consider any agent $e$ who is part of the winning set $S$
with respect to \BIDVEC.
If $e$'s bid decreases by $\epsilon$, the cost of $S$
decreases by $\epsilon$, while the costs of all other sets decrease by
at most $\epsilon$. Thus, because ties are broken consistently, 
$S$ will still be selected.
\end{proof}

\Omit{
\subsubsection{Limitations of Composable Mechanisms}

As we will show for both flows and cuts, implementation of  \SELRULE may require solving an optimization problem.  Therefore in order to achieve a polynomial-time mechanism, we need to deal with two problems: a polynomial-time implementation of the \SELRULE and running the \EVMECH mechanism in polynomial time on the reduced instance. In this paper, we will achieve both for both flows and cuts. However, one natural question is whether approximation techniques could be applied to deal with NP-complete selection rules. The following theorem links the hardness of approximation of a selection rule to the existence of any approximate composable mechanism using that selection rule.

\begin{theorem}
Suppose \SELRULE is APX-hard. Let $\MECH_{\SELRULE}$ be the set of all remove-cover mechanisms with a selection rule that best approximates \SELRULE.  Then $\MECH_{\SELRULE} = \emptyset$ unless P=NP.
\end{theorem}
\begin{proof}
The proof is by contradiction. Let $c > 1$ be such that there is no algorithm that approximate $\SELRULE$ better than a factor of $c$. Suppose $\MECH' \in \MECH_{\SELRULE}$ with $\SELRULE'$ selection rule. Let $\SetSys$ be the set system consisting of elements \Elts along with their bids represent a tight instance for $\MECH'$ i.e. the total bids of surviving set system $\SetSysP$ is $c \sum_{e \in \Elts_{\SELRULE} } \Bid{e}$ where $\Elts_{\SELRULE}$ are the agents selected by \SELRULE. Let $e$ be and agent with positive bid in $\Elts' \setminus \Elts_{\SELRULE}$. (Note that such an agent exists.) Then, by composability we have: $\SelRule{\BidP{e}, \BIDVECEXCL{e}}' = \Elts'$ where $\BidP{e} = \Bid{e} - \epsilon$ for some $\epsilon > 0$. It is sufficient to show that $\SelRule{\BidP{e}, \BIDVECEXCL{e}} = \Elts_{\SELRULE}$, then we will derive a contradiction since $\SELRULE'$ approximates \Elts for a factor better than $c$.

In order to establish the desired inequality, we will make sure that the \SELRULE optimization problem has a unique optimal solution on \SetSys. \emph{Here is the point that I am not sure how to resolve. If we take an edge in an arbitrary OPT which is not part of our $\Elts'$ and decrease the bid by some small $\epsilon'$, it is not clear whether we will preserve the tightness of our instance. Another approach is to increase the bid of all but one optimal solution. But, in that case, we may end up raising the bid of some agent in $\Elts'$ as well. This again may ruin the tightness of our instance. It looks like if for all $e \in \Elts'$ $e$ is part of some optimal solution we are stuck. What do you think?}

\end{proof}

}


\section{A Mechanism for Flows} \label{sec:flows}
We apply the methodology of Theorem
\ref{thm:reduction} to design a mechanism \FLOWMECH for purchasing $k$
edge-disjoint $s$-$t$ paths. We are given a (directed) graph
$G=(V,E)$, source $s$, sink $t$, and target number $k$.
As discussed earlier, the agents are edges of $G$.
We assume that $G$ is monopoly-free, which is equivalent to saying
that the minimum $s$-$t$ cut contains at least $k+1$ edges.
For convenience, we will refer to a set of $k$ edge-disjoint $s$-$t$ paths
simply as a $k$-flow, and omit $s$ and $t$.

To specify \FLOWMECH, all we need to do is describe a composable
pre-processing rule \SELRULE. Our rule is simple:
Choose $(k+1)$ edge-disjoint $s$-$t$ paths, of minimum total bid with
respect to \BIDVEC, breaking ties lexicographically.
We call such a subgraph a \todef{$(k+1)$-flow}, where it is implicit
that we are only interested in integer flows, and identify the flow
with its edge set.
Call the minimum-cost $(k+1)$-flow \MING.
(In Section \ref{sec:strongly-monotone}, we generically referred to
this set system as \SetSysP.)

\begin{theorem} \label{thm:flowmech}
The mechanism \FLOWMECH is truthful and $2(k+1)$-competitive and runs
in polynomial time.
\end{theorem}

We show this theorem in three parts. First, we establish that the
$k$-flow problem on \MING indeed forms a Vertex Cover instance
(Lemma \ref{lem:vertex-cover-flow-equivalence}).
By far the most difficult step is showing that the lower bound
satisfies
$\LNash[\MING]{\COSTVEC} \leq (k+1) \cdot \LNash[G]{\COSTVEC}$
for all cost vectors \COSTVEC (Lemma \ref{lem:nash-comparison}).
The composability of \SELRULE follows from Lemma \ref{lem:composable}.
Together, these three facts allow us to apply Theorem
\ref{thm:reduction}, and conclude that \FLOWMECH is a truthful
$2(k+1)$-competitive mechanism.
Finally, we verify that \FLOWMECH runs in polynomial time
(Lemma \ref{lem:flow-polynomial}).

\begin{lemma} \label{lem:vertex-cover-flow-equivalence}
The instance \SetSysP whose feasible sets are all $k$-flows on 
\MING is a Vertex Cover set system.
\end{lemma}

\begin{proof}
Recall that \MING is a minimal $(k+1)$-flow, a fact that we exploit
repeatedly in this proof.
The edges of \MING are the vertices in the Vertex Cover instance.
For clarity, consider explicitly the graph \VCG, which contains a
vertex $u_e$ for each edge $e \in \MING$, and an edge between
$u_e, u_{e'}$ if and only if removing $e$ would create a monopoly for
$e'$. This is the case iff there exists at least one minimum $s$-$t$
cut in \MING containing both $e$ and $e'$; 
in particular, \VCG is symmetric.
The construction is illustrated for the case $k=2$
in Figure \ref{fig:flow-to-cover}.
An alternative characterization of the edges in \VCG is given in
Proposition \ref{prop:cut-reachability} below, and will be used as part of this proof.
For any set of edges $E'$ in \MING, let \Nodes{E'} be the
corresponding set of nodes in \VCG.
Thus, for any minimum $s$-$t$ cut $E'$, the set \Nodes{E'} forms a
clique in \VCG.

\begin{figure}[htb]
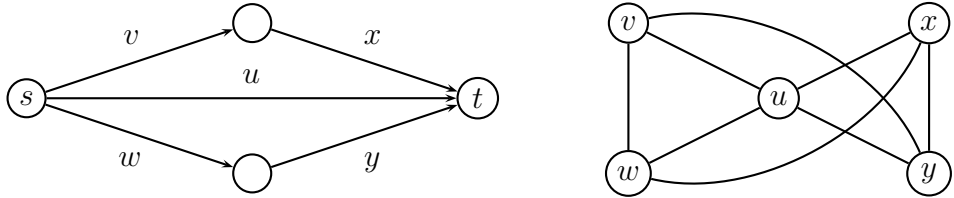

\begin{center}
\psset{unit=1.0cm} \pspicture(-3.5,-1)(9.5,1.5)
\cnodeput(-3,0){s}{$s$} \cnodeput(3,0){t}{$t$}
\cnodeput(0,1){a}{\phantom{$s$}} \cnodeput(0,-1){b}{\phantom{$s$}}

\ncline{->}{s}{t}\Aput{$u$} 
\ncline{->}{s}{a}\Aput{$v$} 
\ncline{->}{s}{b}\Bput{$w$} 
\ncline{->}{a}{t}\Aput{$x$}
\ncline{->}{b}{t}\Bput{$y$}

\cnodeput(7,0){u}{$u$} \cnodeput(5,1){v}{$v$} \cnodeput(5,-1){w}{$w$}
\cnodeput(9,1){x}{$x$} \cnodeput(9,-1){y}{$y$}

\ncline{-}{u}{v} \ncline{-}{u}{w} \ncline{-}{u}{x} \ncline{-}{u}{y} 
\ncline{-}{v}{w} \ncline{-}{x}{y} 
\ncarc[arcangleA=40,arcangleB=40]{-}{v}{y}
\ncarc[arcangleA=-40,arcangleB=-40]{-}{w}{x}
\endpspicture
\caption{A minimally 3-connected graph (left) and the resulting vertex cover
  instance for $k=2$ (right).} \label{fig:flow-to-cover}
\end{center}
\end{figure}

If $E'$ is a $k$-flow, then for any pair of edges $e,e'$ that lie on
a minimum $s$-$t$ cut, $E'$ must contain at least one of $e,e'$.
Thus, \Nodes{E'} is a vertex cover of \VCG.

Conversely, let $E'$ be a set of edges in \MING such that \Nodes{E'}
is a vertex cover of \VCG. We will show that for every $s$-$t$ cut
$F \subseteq E$, at least $k$ edges of $E'$ cross $F$, i.e.,
$\SetCard{E' \cap F} \geq k$. This will imply that $E'$ is a $k$-flow.
Assume for contradiction that $\SetCard{E' \cap F} < k$.
Because \Nodes{E'} is a vertex cover of \VCG, there can be no edge
between any pair of vertices in \Nodes{F \setminus E'} in \VCG. By definition,
this means that for any pair $e, e' \in F \setminus E'$, there is no
minimum $s$-$t$ cut containing both $e$ and $e'$.
By Proposition \ref{prop:cut-reachability} below, this is equivalent
to saying that for each pair $e, e' \in F \setminus E'$, the graph
\MING contains a path from $e$ to $e'$ or a path from $e'$ to $e$.

Consider a directed graph whose vertices are the edges
$F \setminus E'$, with an edge from $e$ to $e'$ whenever \MING
contains a path from $e$ to $e'$. By the above argument, this graph is
a tournament graph, and thus contains a Hamiltonian path. That is,
there is an ordering $e_1, \ldots, e_\ell$ of the edges in $F \setminus E'$
such that each $e_{i+1}$ is reachable from $e_i$ in \MING. By adding a
path from $s$ to $e_1$ and from $e_\ell$ to $t$, we thus obtain an
$s$-$t$ path $P$ containing all edges in $F \setminus E'$.
The graph $\MING \setminus P$ is a $k$-flow, so the set
$E' \cap F$, having size less than $k$, cannot be an $s$-$t$ cut in
$\MING \setminus P$. Let $P'$ be an $s$-$t$ path in $\MING \setminus P$
disjoint from $E' \cap F$. By construction, $P'$ is also disjoint from
$F \setminus E'$. Thus, we have found an $s$-$t$ path $P'$ in \MING
disjoint from $F$, contradicting the assumption that $F$ is an $s$-$t$
cut.
\end{proof}

\begin{proposition} \label{prop:cut-reachability}
Let \MING be a graph consisting of $k+1$ edge-disjoint $s$-$t$ paths,
and let $e=(u,v),e'=(u',v')$ be two edges of \MING.
Then, there is a minimum $s$-$t$ cut containing both $e$ and $e'$ if
and only if there is no path from $v$ to $u'$ and no path from $v'$ to
$u$.
\end{proposition}

\begin{proof}
Assume that there is a path from $v$ to $u'$. Let $P$ be a
concatenation of an $s$-$v$ path using $e$, the path from $v$ to $u'$, and a
path from $u'$ to $t$ using $e'$.
Then, $\MING \setminus P$ is a $k$-flow, and therefore has $k$
edge-disjoint paths. Any $s$-$t$ cut in \MING must thus contain at
least $k$ edges from $\MING \setminus P$, and no $s$-$t$ cut with
fewer than $k+2$ edges can contain both $e$ and $e'$.

Conversely, if there is no minimum cut containing both $e,e'$, then
every minimum cut in $\MING \setminus \SET{e,e'}$ must contain
$k$ edges. Thus, $\MING \setminus \SET{e,e'}$ contains $k$
edge-disjoint $s$-$t$ paths. Removing these paths from \MING leaves us
with a 1-flow, i.e., one $s$-$t$ path. By construction, this path must
contain $e$ and $e'$; thus, at least one is reachable from the other.
\end{proof}

\Omit{
\begin{lemma} \label{lem:flow-truthful}
The selection rule is composable.
\end{lemma}

\begin{proof}
Consider any agent $e$ who is part of the cheapest $(k+1)$-flow
\MING with respect to \COSTVEC.
If $e$'s cost decreases by $\epsilon$, the cost of \MING
decreases by $\epsilon$, while the cost of all other sets decrease by
at most $\epsilon$. Thus, so long as ties are broken consistently, the
same set \MING will be selected\footnote{Note that this proof in fact
  shows that any pre-processing rule simply selecting the cheapest set
with a certain property $\mathcal{P}$ is composable.}.
\end{proof}
}

\begin{lemma} \label{lem:nash-comparison}
$\LNash[\MING]{\COSTVEC} \leq (k+1) \cdot \LNash[G]{\COSTVEC}$
for all 
\COSTVEC.
\end{lemma}

\begin{proof}
Let $S$ be the cheapest $k$-flow in $G$ with respect to the costs
\COSTVEC.
Because \MING is a $(k+1)$-flow, Corollary \ref{cor:lnash-snash} below
implies that $\LNash[\MING]{\COSTVEC} = k \cdot \MPC[\MING]{\COSTVEC}$,
where \MPC[\MING]{\COSTVEC} is the cost of the most expensive $s$-$t$
path in \MING.

Let \NBIDVEC be a solution to the LP~\eqref{eqn:nash-def} 
with cost vector \COSTVEC on the graph $G$.
Define a graph \TG consisting of all edges in $S$, as well as all
edges that are in at least one tight feasible set $T$ (i.e., a set $T$
for which the constraint (iii) is tight,
meaning that $\NBidSum{T} = \NBidSum{S}$).

By definition, \TG contains at least $k+1$ edge-disjoint $s$-$t$ paths.
Lemma \ref{lem:equal-path-lengths} below (the key step) implies that all
$s$-$t$ paths in \TG have the same total bid with respect to \NBIDVEC.
Let $P$ be an $s$-$t$ path in \TG of maximum total cost \CostSum{P}.
By individual rationality (Constraint (i) in the
LP~\eqref{eqn:nash-def}), we have that
$\NBidSum{P} \geq \CostSum{P}$, and hence
$\NBidSum{P'} \geq \CostSum{P}$ for all $s$-$t$ paths $P'$.
In particular,
$\LNash[G]{\COSTVEC} = \NBidSum{S} \geq k \cdot \CostSum{P}$.
Since \MING is a minimum-cost $(k+1)$-flow, 
and because $\TG$ contains at least $k+1$
edge-disjoint $s$-$t$ paths, we have 
$\MPC[\MING]{\COSTVEC} \le (k+1) \,\CostSum{P}$.  
Thus,
\[ 
\LNash[G]{\COSTVEC} 
\; \geq \; k \cdot \CostSum{P}
\; \geq \; \frac{k}{k+1} \cdot \MPC[\MING]{\COSTVEC}
\; = \; \frac{1}{k+1} \cdot \LNash[\MING]{\COSTVEC},
\] 
which completes the proof.
\end{proof}

\begin{lemma} \label{lem:equal-path-lengths}
Let \NBIDVEC be a solution to the LP~\eqref{eqn:nash-def}, 
and \TG as in the proof of Lemma
\ref{lem:nash-comparison}. Let $v$ be an arbitrary node in $G$, and
$P_1,P_2$ two paths from $v$ to $t$.
Then, $\NBidSum{P_1} = \NBidSum{P_2}$.
\end{lemma}

\begin{emptyproof}
Let \FLOWCOLL be the collection of all tight $k$-flows from $s$ to $t$
except $S$, i.e., the set of all $\FLOW$ such that $\FLOW \neq S$, $\FLOW$
consists of exactly $k$ edge-disjoint $s$-$t$ paths, and
$\NBidSum{\FLOW} = \NBidSum{S}$.
We define a directed multigraph \MG as follows: for each
$\FLOW \in \FLOWCOLL$, we add to \MG a copy of each edge $e \in \FLOW$
(creating duplicate copies of edges $e$ which are in multiple flows
$\FLOW$). We call these edges \todef{forward edges}.
In addition, for each edge $e=(u,v) \in S$, we add
$\SetCard{\FLOWCOLL}$ copies of the \todef{backward edge}
$(v,u)$ to \MG, i.e., we direct $e$ the other way.

In the resulting multigraph,
each node $v$ has an in-degree equal to its out-degree.
For $v \neq s, t$ this follows since each edge set we added constitutes a flow. 
For $v=s,t$, it follows since each $\FLOW \in
\FLOWCOLL$ adds $k$ edges out of $s$ and into $t$, while the
$\SetCard{\FLOWCOLL}$ copies of $S$ add $k \SetCard{\FLOWCOLL}$ edges
into $s$ and out of $t$.  
As a result, \MG\ is Eulerian, a fact we use below.

We define a mapping \EIm{e}, which assigns to each edge $e \in \MG$ its
``original'' edge in $G$. As usual, we extend notation and write
$\EIm{R} = \Set{\EIm{e}}{e \in R}$ for any set $R$ of edges.

We will be particularly interested in analyzing collections of cycles
in \MG. We say that two cycles $\Cyc{1}, \Cyc{2}$ are
\todef{image-disjoint} if $\EIm{\Cyc{1}} \cap \EIm{\Cyc{2}} = \emptyset$.
A \todef{cycle set} is any set of zero or more image-disjoint cycles in
\MG (which we identify with its edge set),
and \CYSETCOLL denotes the collection of all cycle sets.
For a cycle set $\CYSET \in \CYSETCOLL$, let \CYSETF and \CYSETB
denote the set of forward and backward edges in \CYSET, respectively.
Then, we define
$\CYImage{\CYSET} = S \cup \EIm{\CYSETF} \setminus \EIm{\CYSETB}$.
It is easy to see that for each cycle set \CYSET,
\CYImage{\CYSET} is a $k$-flow in \TG.
Conversely, for every $k$-flow $\FLOW$ in \TG, there is a cycle set
$\CYSET \in \CYSETCOLL$ with $\CYImage{\CYSET} = \FLOW$.

We assign each edge $e \in \MG$ a weight \Weight{e}.
For forward edges $e$, we set $\Weight{e} = \NBid{\EIm{e}}$, while for
backward edges $e=(v,u)$, we set $\Weight{e} = -\NBid{\EIm{e}}$.
Notice that because each copy of $S$ contributes weight
$-\NBidSum{S}$, and each set $\FLOW \in \FLOWCOLL$ contributes
$\NBidSum{\FLOW} = \NBidSum{S}$, the sum of all weights in \MG is 0.

Now, let \CYSET be any cycle set, and $\FLOW = \CYImage{\CYSET}$ its
corresponding $k$-flow. We have
\[ 
\sum_{e \in \CYSET} \Weight{e}
\; = \; \NBidSum{\FLOW \setminus S} - \NBidSum{S \setminus \FLOW}
\; = \; \NBidSum{\FLOW} - \NBidSum{S},
\] 
Thus $\FLOW$ is tight, i.e., $\NBidSum{\FLOW} = \NBidSum{S}$, if and only if
$\sum_{e \in \CYSET} \Weight{e} = 0$.


We next show that for any cycle \CYC in \MG, the $k$-flow
\CYImage{\CYC} is tight, and therefore that $C$ has total weight zero, 
i.e., $\sum_{e \in \CYC} w_e=0$.
Assume for contradiction that this is not the case, and let \CYC be a cycle with
$\sum_{e \in \CYC} \Weight{e} \neq 0$. Let $\FLOW = \CYImage{\CYC}$ be
the corresponding $k$-flow.
Because we showed above that
$\sum_{e \in \CYC} \Weight{e} = \NBidSum{\FLOW} - \NBidSum{S}$,
we can rule out the possibility that $\sum_{e \in \CYC} \Weight{e} < 0$;
otherwise, $\NBidSum{\FLOW} < \NBidSum{S}$, which would violate Constraint
(iii) of the LP~\eqref{eqn:nash-def}.


If $\sum_{e \in \CYC} \Weight{e} > 0$, consider the multigraph $\MG'$ 
obtained by removing \CYC from \MG. Its total weight is
$\sum_{e \notin \CYC} \Weight{e} < 0$, because the sum of all
weights in \MG is 0 (as shown above).
Since \MG is Eulerian, so is $\MG'$,
and its edges can be partitioned into
a collection of edge-disjoint cycles
$\SET{\Cyc{1}, \ldots, \Cyc{\ell}}$. By the Pigeonhole
Principle, at least one of the \Cyc{i} must have negative total weight.  
But then $\NBidSum{\FLOW_i} < \NBidSum{S}$, where 
$\FLOW_i = \CYImage{\Cyc{i}}$, violating Constraint (iii) of~\eqref{eqn:nash-def} 
as in the previous case.    
This completes the proof that $\CYImage{\CYC}$ is tight for any cycle \CYC.  
By our observation above, the total weight of any cycle \CYC is zero.


Finally, we prove the statement of the lemma by
induction on a reverse topological sorting of the vertices $v$---that
is, an ordering in which the index of $v$ is at least as large as the index of any
$u$ such that $(v,u) \in \TG$. Because \TG is acyclic, such a
sorting exists. The base case $v=t$ is 
trivial. For $v \neq t$, let $P_1, P_2$ be two $v$-$t$ paths.
We distinguish three cases, based on the first edges $e_1=(v,u_1),
e_2=(v,u_2)$  of the paths $P_1, P_2$.

\begin{enumerate}
\item If \MG contains a forward edge $(v,u_1)$ and a backward edge
  $(u_2,v)$ (or vice versa), then since every set of edges added to \MG is a flow,
  \MG must contain a $v$-$t$ path $P'_1$ entirely consisting of
  forward edges and starting with $e_1$,
  and a $t$-$v$ path $P'_2$ entirely consisting of
  backward edges and ending with $e_2$ (backward).
  Applying the induction hypothesis to $u_1$ and $u_2$, 
  since $P_1$ and $P'_1$ share their first edges and similarly for 
  $P_2$ and $P'_2$, we have 
  $\NBidSum{\EIm{P'_1}} = \NBidSum{P_1}$ and
  $\NBidSum{\EIm{P'_2}} = \NBidSum{P_2}$.
  Because $P'_1 \cup P'_2$ forms a cycle, it has total weight 0.  
  Then $\NBidSum{\EIm{P'_2}} = -\Weight{P'_2} = \Weight{P'_1} =
  \NBidSum{\EIm{P'_1}}$, and so $\NBidSum{P_1} = \NBidSum{P_2}$.
  
\item If \MG contains forward edges $(v,u_1)$ and $(v,u_2)$, then it
  contains $v$-$t$ paths $P'_1, P'_2$ starting with $(v,u_1)$
and $(v,u_2)$, respectively, and consisting entirely of forward edges. 
Applying the induction hypothesis to $u_1$ and $u_2$, we have
  $\NBidSum{\EIm{P'_1}} = \NBidSum{P_1}$ and
  $\NBidSum{\EIm{P'_2}} = \NBidSum{P_2}$.
  Since every set of edges added to \MG  is a flow,
  \MG must contain an
  $s$-$v$ path $P$ consisting entirely of forward edges, and 
  \MG must also contain a $t$-$s$ path $P'$
  consisting entirely of backward edges.
  Because $P \cup P' \cup P'_i$ forms a cycle for each $i \in \{1,2\}$ 
  and has total weight zero,
  we obtain $\NBidSum{P_i} = \NBidSum{\EIm{P'_i}} = -\Weight{P \cup P'}$
  for each $i$.  In particular, $\NBidSum{P_1} = \NBidSum{P_2}$.
  
\item Finally, if \MG contains backward edges $(u_1,v)$ and $(u_2,v)$,
  we apply an argument similar to the previous case.
  By induction,
  $\NBidSum{\EIm{P'_1}} = \NBidSum{P_1}$, and
  $\NBidSum{\EIm{P'_2}} = \NBidSum{P_2}$.  
  Again using the fact that \MG\ consists of flows, it contains $t$-$v$ paths $P'_1, P'_2$ 
  with respective last edges $(u_1,v)$ and $(u_2,v)$.  In addition, \MG contains a $v$-$s$ path $P$
  consisting entirely of backward edges, and an $s$-$v$ path $P'$
  consisting entirely of forward edges. Then for each $i \in \{1,2\}$, 
  $P \cup P' \cup P'_i$ forms a cycle with total weight zero, so
  $\NBidSum{P_1} = \NBidSum{P_2}$.\QED
\end{enumerate}
\end{emptyproof}

As a corollary, we can derive a characterization of Nash Equilibria in
$(k+1)$-flows.

\begin{corollary} \label{cor:lnash-snash}
If $G$ is a $(k+1)$-flow, then a bid vector \NBIDVEC is a Nash Equilibrium
if and only if $\NBidSum{P} = \MPC[G]{\COSTVEC}$ for all $s$-$t$ paths $P$.
In particular, all Nash Equilibria have the same total cost
$\NBidSum{S} = k \cdot \MPC[G]{\COSTVEC}$, where $S$ is the winning set.
\end{corollary}

\begin{proof}
First, because $G$ is a $(k+1)$-flow, the graph \TG\  
defined in the proof of Lemma~\ref{lem:nash-comparison}
actually equals $G$, since it must contain $k+1$ edge-disjoint $s$-$t$
paths. If \NBIDVEC is a Nash Equilibrium, then
by Lemma \ref{lem:equal-path-lengths}, all $s$-$t$ paths $P$ have the same
total bid \NBidSum{P}.
Let $\hat{P}$ be an $s$-$t$ path maximizing \CostSum{P}, i.e.,
$\CostSum{\hat{P}} = \MPC[G]{\COSTVEC}$.
$G \setminus \hat{P}$ is a $k$-flow, and clearly the cheapest $k$-flow by
definition of $\hat{P}$. Therefore, all agents in $\hat{P}$ lose, and
$\NBidSum{\hat{P}} = \CostSum{\hat{P}}$ by Constraint
(ii) of the LP~\eqref{eqn:nash-def}.
\end{proof}

Finally, we show that the mechanism \EVMECH runs in polynomial time
for the special case of graphs derived from $k$-flows.

\begin{lemma} \label{lem:flow-polynomial}
For the Vertex Cover instance derived from computing a $k$-flow on
a $(k+1)$-flow, the mechanism \EVMECH runs in polynomial time.
\end{lemma}

\begin{proof}
There are two steps which are of concern:
computing the values \Tot{v}, and finding the cheapest vertex cover
with respect to the scaled bids. The latter is exactly a Minimum Cost
Flow problem by Proposition \ref{lem:vertex-cover-flow-equivalence},
and thus solvable in polynomial time with standard algorithms
\cite{ahuja:magnanti:orlin}. 

For the former, we claim that
$\Tot{u_e} = k$ for all $u_e \in \VCG$. By Proposition
\ref{prop:clique} and LP duality, \Tot{u_e} is upper bounded by the
chromatic number of $u_e$'s neighborhood, and lower bounded by its
clique number. Since each edge $e \in \MING$ is part of a minimum cut
of size $k+1$, and the edges of the minimum cut form a clique in \VCG,
the clique number of $u_e$'s neighborhood is at least $k$. 
On the other hand, we can 
decompose \MING into $k+1$ edge-disjoint paths, and color 
the vertices corresponding to each path with its own color in \VCG.
By Proposition \ref{prop:cut-reachability}, this is a valid
coloring, and shows that $u_e$'s neighborhood is $k$-colorable.
\end{proof}

\begin{remark} \label{rem:flows-tight}
The factor 2 in the result of Theorem \ref{thm:flowmech} comes from
the factor \half in the lower bound in Lemma \ref{lem:vertex-lower}.
Using a more refined lower bound based on Young's Inequality, Chen et
al.~\cite{chen:elkind:gravin:petrov} showed that for an unscaled
version of the Vertex Cover mechanism, the factor \half in the lower bound
is unnecessary. 
For the instances of Vertex Cover produced as a result of the pruning
in this section, the mechanism from \cite{chen:elkind:gravin:petrov}
coincides with \EVMECH, and hence \FLOWMECH is the same as 
the flow mechanim \cite{chen:elkind:gravin:petrov}.

Chen et al.~also showed that while \FLOWMECH is
$(k+1)$-competitive when compared against the buyer-optimal lower
bound \cite{BeyondVCG}, it is in fact optimal compared to the
buyer-pessimal version \cite{elkind:goldberg:goldberg:frugality}. 
\end{remark}


\section{A Mechanism for Cuts}

As a second application of our methodology, we give a competitive
mechanism \CUTMECH for purchasing an $s$-$t$ cut, given a (directed)
graph $G=(V,E)$, source $s$, and sink $t$.
Again, the agents are edges.
Here, the necessary monopoly-freeness is equivalent to $G$ not
containing the edge $(s,t)$.

As before, it suffices to specify and analyze a
composable pre-processing rule \SELRULE.
Our pre-processing rule is to compute a minimum-cost set $E'$ of edges
(with respect to the submitted bids \BIDVEC),
such that $E'$ contains at least two edges from each $s$-$t$ path.
We call such an edge set a \todef{double cut}.  
We show below restricting the set system to $E'$ gives a Vertex Cover instance, 
and at most increases the cost of the winning set by a factor of $2$.

\subsection{Restricting to a double cut}

To restrict the set system to $E'$, we contract all edges in $E \setminus E'$. 
Since no such edge will be cut, contracting it ensures that  
its endpoints will always lie on the same side of the cut. 
Let \MINCG denote the resulting graph.
We begin with a simple structural lemma about \MINCG.

\begin{lemma} \label{lem:length2}
In \MINCG, all $s$-$t$ paths have length exactly 2.
\end{lemma}

\begin{proof}
If there were an $s$-$t$ path of length 1 in \MINCG, i.e., an edge
$(s,t)$, then consider the edge $(u,v)$ in the original graph
corresponding to $(s,t)$. Because $u$ was contracted with $s$, and $v$
with $t$, there must be an $s$-$u$ path and a $v$-$t$ path in $G$
using only edges from $E \setminus E'$.  
In that case, $(u,v)$ is the only edge on this path contained in $E'$, so 
$E'$ cannot have been a double cut.
Similarly, if there were an $s$-$t$ path $P$ of length at least 3, then at
least one edge $(u,v)$ of $P$ has neither $s$ nor $t$ as an
endpoint. This edge could be safely contracted, i.e., removed from $E'$, 
in which case $E'$ was not a minimum-cost double cut.
\end{proof}

\begin{theorem} \label{thm:cutmech-comp}
The double cut selection rule is composable and produces a Vertex
Cover instance with
$\LNash[\MINCG]{\COSTVEC} \leq 2 \LNash[G]{\COSTVEC}$.
Furthermore, both the selection rule and the subsequent Vertex Cover
mechanism can be computed in polynomial time.
Thus, \CUTMECH is a polynomial-time 4-competitive mechanism.
\end{theorem}

\noindent
Composability follows from Lemma \ref{lem:composable}, and
the final conclusion then follows from Theorem \ref{thm:reduction}
once we establish the other claims.

We can obtain a Vertex Cover instance by imposing a graph structure on \MINCG, 
treating each edge as a vertex and adding an edge $(e,e')$ between any two edges 
that form an $s$-$t$ path.
A set of edges is an $s$-$t$ cut if and only if it contains at least one of $e, e'$ 
in each such pair, so it is a vertex cover of the resulting graph.

We can think of this in turn as a flow problem as follows. Lemma
\ref{lem:length2} implies that \MINCG is of the following form: in
addition to $s$ and $t$, there are vertices $v_1, \ldots, v_{\ell}$,
and for each $i=1, \ldots, \ell$, a set of parallel edges $E_i$ from
$s$ to $v_i$, and a set of parallel edges $E'_i$ from $v_i$ to $t$.
Any $s$-$t$ cut has to include, for each $i$, all of $E_i$ or all of
$E'_i$.  Thus, if we define a minimally
2-connected graph consisting of a series of vertices $u_0, u_1, \ldots, u_\ell$ with
two vertex-disjoint paths of length $\SetCard{E_i}$ and $\SetCard{E'_i}$ between 
$u_{i-1}$ and $u_{i}$ for each $i$, an $s$-$t$ cut in \MINCG is a 1-flow from $u_0$ to $u_\ell$.
 We can then apply Lemma \ref{lem:vertex-cover-flow-equivalence}. Notice that this
equivalence also establishes that \EVMECH runs in polynomial time on
the instances produced by this selection rule.

As before, the key part is to analyze the increase in the lower bound.
\begin{lemma}
For all cost vectors \COSTVEC,
$\LNash[\MINCG]{\COSTVEC} \leq 2 \LNash[G]{\COSTVEC}$.
\end{lemma}

\begin{proof}
Let $(S, \Compl{S})$ be the cheapest $s$-$t$ cut in $G$
with respect to the costs \COSTVEC, and
\NBIDVEC a solution to the LP~\eqref{eqn:nash-def} with cost vector \COSTVEC on the graph $G$.
Let \ALLCUTS be the set of all minimum $s$-$t$ cuts $(T, \Compl{T})$
with respect to the costs \NBIDVEC; thus, each of these cuts has cost
\NBidSum{E(S, \Compl{S})}.
Define $T^- = \bigcap_{(T, \Compl{T}) \in \ALLCUTS} T$,
and $T^+ = \bigcup_{(T, \Compl{T}) \in \ALLCUTS} T$.
Then, both $(T^-, \Compl{T^-})$ and $(T^+, \Compl{T^+})$
are minimum $s$-$t$ cuts as well (see, e.g.,
\cite[Exercise 6.39]{ahuja:magnanti:orlin}).

Furthermore, the edge sets $E(T^-, \Compl{T^-})$
and $E(T^+, \Compl{T^+})$ are disjoint. For assume that there is an
edge $e=(u,v)$ in common between these sets.
Then, $u \in \bigcap_{(T, \Compl{T}) \in \ALLCUTS} T$ and
$v \in \bigcap_{(T, \Compl{T}) \in \ALLCUTS} \Compl{T}$.
In particular, this implies that $u \in S$ and $v \in \Compl{S}$. \
As stated above in equation~\eqref{eq:tight}, 
since \NBIDVEC maximizes the LP~\eqref{eqn:nash-def},
Constraint (iii) must be tight for some feasible set excluding $e$, since 
otherwise the bid \NBid{e} could be increased.  
Let $(T, \Compl{T})$ be the corresponding cut.
Then $\NBidSum{E(T, \Compl{T})} = \NBidSum{E(S, \Compl{S})}$,
and $e$ does not cross $(T, \Compl{T})$.  
Thus, either both $u$ and $v$ are in $T$, or both are in \Compl{T}.
Since $(T, \Compl{T}) \in \ALLCUTS$, this gives a contradiction.


Now define
$\TG := E(T^-, \Compl{T^-}) \cup E(T^+, \Compl{T^+})$.
Because \TG consists of two disjoint $s$-$t$ cuts, it is a double cut,
and the cost-minimality of
\MINCG implies that $\CostSum{\TG} \geq \CostSum{\MINCG}$.  
These two cuts both have minimal cost, so $\LNash[G]{\COSTVEC}
= \NBidSum{\TG}/2$.
By the ``individual rationality'' LP constraint (i),
$\NBidSum{\TG} \geq \CostSum{\TG}$, and hence
\[ 
\LNash[G]{\COSTVEC}
\;=\; \frac{\NBidSum{\TG}}{2}
\;\geq\; \frac{\CostSum{\TG}}{2}
\;\geq\; \frac{\CostSum{\MINCG}}{2}
\;\geq\; \frac{\LNash[\MINCG]{\COSTVEC}}{2}.
\]
For the last inequality, notice that in the ``Nash Equilibrium'' on
\MINCG, for each $i$, the cheaper of $E_i$ and $E'_i$ will
collectively raise their bids to the cost of the more expensive one,
so the total bid of the winning set will be $\LNash[\MINCG]{\COSTVEC} = 
\sum_i \max(\CostSum{E_i}, \CostSum{E'_i}) \leq \CostSum{\MINCG}$.
\end{proof}

\subsection{A Primal-Dual Algorithm for Minimum Double-Cuts}

Finally, we present a polynomial time algorithm to compute a
minimum-cost double cut.
The minimum-cost double cut is characterized by integer solutions to
the following LP, where \ALLPATH denotes the set of all $s$-$t$ paths in $G$.

\begin{LP}[LP:min-2-cut]{Minimize}{\sum_{e\in E} \Cost{e} x_e}
\sum_{e \in P} x_e \geq 2 & \mbox{ for all } P \in \ALLPATH\\
x_e \leq 1 & \mbox { for all edges } e \in E\\
x_e \geq 0 & \mbox{ for all } e,
\end{LP}

\begin{remark}
It is not difficult to show that the constraint matrix for this LP is
totally unimodular.
By a well-known theorem \cite{papadimitriou:steiglitz},
because the right-hand sides of the constraints are integral,
total unimodularity implies that all the vertices of the LP's polytope
are integral.
Since there is a separation oracle for the LP (as well as an
equivalent polynomial-sized LP formulation), 
an integer solution can
be found in polynomial time, giving us a polynomial-time algorithm.  
However, the resulting algorithm is rather inefficient.
\end{remark}

Here, we present a more efficient primal-dual algorithm
generalizing the Ford-Fulkerson Max-Flow algorithm.
The dual of the LP is

\begin{LP}[LP:min-2-cut-dual]{Maximize}{%
2 \sum_{P \in \ALLPATH} f_P - \sum_e r_e}
\sum_{P: e \in P} f_P \leq c_e + r_e & \mbox{ for all $e \in E$} \\
f_P, r_e \geq 0 & \mbox{ for all $P \in \ALLPATH$ and all $e \in E$} \, .
\end{LP}

\noindent
We interpret the dual variables $f_P$ as describing a flow in the usual way.  
That is, the flow along each edge $e$ is 
\[
f_e = \sum_{P: e \in P} f_P \, .
\]
We say that $e$ is saturated if $f_e = c_e+r_e$.  
We call $r_e$ the \todef{relief} on $e$: 
in order to send more flow on an edge $e$, we can increase its capacity, 
but we pay for it in the objective
function.  It is worth augmenting the flow along a path so long as at
most one edge on the path is saturated, since increasing the flow and the
relief of the saturated edge at the same time increases the dual
objective.

Our primal-dual algorithm is similar to the Ford-Fulkerson algorithm,
and is based on the same concept of a residual graph. The residual graph
contains \todef{forward edges} for all edges $e$ in the original
graph, \emph{even when they are saturated}, because it is possible to
send more flow by adding relief. In addition, if $e = (u,v)$
carries flow $f_e$, then the residual graph, as usual, contains the \todef{backward edge}
$(v,u)$ with capacity $f_e$.
To capture how much relief would have to be added to augment the flow
along a path, we define, for each edge $e$ in the residual graph, a
\todef{length} \Length{e} as follows:
\begin{enumerate}
\item If $e$ is a saturated forward edge, then $\Length{e}=1$.
\item If $(v,u)$ is a backward edge such that $(u,v)$ has
  \emph{positive} relief, then $\Length{(v,u)} = -1$.
\item The lengths of all other edges are $\Length{e} = 0$.
\end{enumerate}
For paths $P$, we define $\PLength{P} = \sum_{e \in P} \Length{e}$.
We give our primal-dual algorithm as Algorithm~\ref{alg:min-2-cut}.

\begin{algorithm}
\caption{Flow computation for Minimum Double Cut} \label{alg:min-2-cut}
\begin{algorithmic}[1]
\STATE {\bf{Flow Computation:}}
\STATE Let $f$ be an arbitrary maximum flow on $G$.
\WHILE{there is an $s$-$t$ path $P$ with $\PLength{P} \leq 1$ in the residual
  graph $G_f$}
\STATE Let $P$ be such a path with minimum length \PLength{P}.
\STATE 
Augment the flow on $P$ by $\delta$, 
while simultaneously increasing the relief of 
any saturated edge by $\delta$, and decreasing the relief of any backward edge by $\delta$, 
for the smallest value of $\delta$ such that this action increases $\PLength{P}$, i.e., 
the smallest $\delta$ such that either a new forward edge becomes saturated or the relief 
on a backward edge becomes zero.
\ENDWHILE
\end{algorithmic}
\end{algorithm}

Notice that for any path $P$ of length at most 1, augmenting the flow
increases the dual objective. This follows since the total number of saturated edges
is at most one greater than the total number of backward edges with
relief; for the latter, each unit of flow reduces the total relief by
one unit, while for the former, each unit of flow increases the total
relief by one unit. Thus, the total increase in relief for sending
$\delta$ units of flow is at most $\delta$, while the first term of
the objective function, i.e., the value of the flow, increases by $2\delta$.

As with the Ford-Fulkerson algorithm, the running time could be
pseudo-polynomial with a poor choice of the augmenting path $P$. But breaking ties for the
smallest total number of edges in $P$ gives strongly polynomial
running time, as with the Edmonds-Karp algorithm.

When the algorithm terminates, we have a set $T$ of saturated edges, and a
subset $R \subseteq T$ of \todef{relief edges} $e$ with $r_e > 0$.
We will pick two edge-disjoint $s$-$t$ cuts $(S_1, \Compl{S_1})$,
$(S_2, \Compl{S_2})$ with the properties that:
\begin{enumerate}
\item Only saturated edges cross either of the cuts.
\item $S_1 \subset S_2$.
\item Each relief edge crosses one of the two chosen cuts.
\end{enumerate}
\noindent
This will naturally satisfy all complementary slackness conditions for
the two LPs, and thus prove optimality of the cuts.


To define and compute the two cuts, we focus on the graph $G'$
obtained from the residual graph by removing all forward edges $e$
with $f_e = 0$. Importantly, we use the same notion of length defined
above. From now on, all references to reachability, distances, etc.~are with
respect to $G'$.

For each node $v$, let \Dist{v} denote the minimum distance from $s$ to
$v$ in $G'$.  
We show next that $G'$ has no negative cycles,
so these distances are well-defined.

\begin{lemma}\label{lem:nonegcyc}
$G'$ has no path from the sink $t$ to the source $s$ of length
strictly less than $-1$. In particular, $G'$ contains no negative
cycles.
\end{lemma}

\begin{proof}
We show by induction that these properties hold for the residual graph
in each iteration. Since $G'$ is obtained from the residual graph only
by deleting edges with zero flow (and thus length 0 or 1 only),
distances can only increase in $G'$.

Initially, all edges have length 0 or 1, and there are no backward
edges, so the claim clearly holds.
If the residual graph contained a negative-length cycle $C$, then $C$ would
have to contain at least one flow-carrying edge $e = (u,v)$. Since $e$ has
incoming flow from $s$ and outgoing flow to $t$, the residual graph
would contain a path of backward edges from $u$ to $s$ and one from $t$ to
$v$. Thus, any negative cycle would give arbitrarily negative-length paths from
$t$ to $s$. It is therefore enough to establish the first claim.

Consider an iteration when flow is augmented along a path $P$.
Suppose that this generates a $t$-$s$ path $P'$ in the residual graph of
length strictly less than $-1$.
$P \cup P'$ gives a cycle.
If we assign edges in $P$ their length prior to the augmentation, and
edges in $P'$ their length after the augmentation,
then the total length of the cycle $P \cup P'$ is negative.
The only edges in $P'$ whose length can have decreased through the
augmentation are the backward versions of edges to which relief was
added by $P$. They were saturated before the augmentation, so their
forward length was 1, and their backward length after augmentation is
$-1$.

Now consider removing edges that appear both forward
and backward in $P \cup P'$. We obtain a union of edge-disjoint
cycles, such that all edges in these cycles were present in the
residual graph prior to the flow augmentation.
Of these edge-disjoint cycles, by the argument of the previous paragraph,
at least one cycle $C$ has negative length with respect
to the previous paragraph's definition. The edges in $P' \setminus P$
don't change their length, so $C$ had negative length before the
augmentation, contradicting the induction hypothesis.
\end{proof}

For two nodes $u,v$, we write $u \dreaches v$ if there is a path of
length at most 0 from $u$ to $v$ in $G'$. We now define the cuts.
Let $S_1 := \Set{v}{\Dist{v} \leq 0}$.
Define $E' := \Set{(u,v) \in R}{\Dist{u} > 0}$.
Now, let $U$ be the set of all vertices lying on a $v$-$t$ path for
some edge $(u,v) \in E'$, and let $\Compl{S_2}$ be the set of all vertices
$y$ such that $y \dreaches w$ for some $w \in U$.
Clearly, $(S_1, \Compl{S_1})$ and $(S_2,\Compl{S_2})$ define two
$s$-$t$ cuts using only saturated edges.

\begin{lemma}\label{lem:csconds}
No edge crosses both cuts $(S_1, \Compl{S_1})$ and $(S_2,\Compl{S_2})$.
Each relief edge $e \in R$ crosses one of the cuts
$(S_1, \Compl{S_1})$ or $(S_2,\Compl{S_2})$.
\end{lemma}

\begin{emptyproof}
To prove the first claim, suppose that $e=(u,v)$ crosses both cuts.
By the definition of $\Compl{S_2}$, that means that there is an edge
$e'=(u',v') \in E' \subseteq R$ such that there is a path of length at
most 0 from $v$ to some node $w$ on a $v'$-$t$ path.
Consider the path from $s$ to $u$ (of length at most 0), followed by
$e$ (of length at most 1), followed by the path from $v$ to $w$, and then the
path to $v'$ backwards, followed by $e'$ backwards. This is a path of
length at most 0 from $s$ to $u'$, meaning that $u'$ should have been in
$S_1$, and contradicting that $e'$ was in $E'$.

To prove the second claim, suppose that a relief edge $e=(u,v)$ crosses neither of the
cuts. We distinguish two cases:
\begin{enumerate}
\item If $u \in S_1$ then, since $e$ does not cross either cut,
  $v \in S_1$, so $\Dist{v} \leq 0$.
  But then the $s$-$v$ path of length at most 0, followed by
  $e$ backwards (of length $-1$), followed by the $u$-$s$ path backward (of length
  at most 0) gives a negative cycle, contradicting Lemma \ref{lem:nonegcyc}.
\item If $u \notin S_1$, then $\Dist{u} > 0$ and $e \in E'$.  Thus $v
  \in U$, and since $v \dreaches v$, we have $v \in \Compl{S_2}$.
  Since $e$ does not cross either cut, we also have $u \in \Compl{S_2}$.
  This means that there is a $w \in U$ and $e'=(u',v') \in E'$ such that $w$
  lies on a $v'$-$t$ path, and $u \dreaches w$.
  Now consider the path from $t$ to $v$ (of length at most 0), then
  using $e$ backwards (of length $-1$), then the length-0 path from $u$ to
  $w$ and the path from $w$ to $v$ backwards (of length at most 0),
  followed by $e'$ backwards, and the path from $u'$ to $s$ backwards
  (of length at most 0). This gives a $t$-$s$ path of total length at most $-2$, 
  again contradicting Lemma \ref{lem:nonegcyc}. \QED
\end{enumerate}
\end{emptyproof}

Lemma~\ref{lem:csconds} implies that the set $R$ of relief edges 
forms a double cut. Thus our algorithm finds a minimum-cost double cut in polynomial time.

\begin{remark}
As in Remark \ref{rem:flows-tight} for the flow mechanism \FLOWMECH,
we can show that on instances derived from the pruning step, \EVMECH
coincides with the mechanism of \cite{chen:elkind:gravin:petrov}.
Thus, the tighter analysis shows that \CUTMECH is in fact
2-competitive. 
We conjecture that \CUTMECH is indeed optimal when compared against
the buyer-pessimal lower bound
of \cite{elkind:goldberg:goldberg:frugality}.
\end{remark}


\section{Directions for Future Work} \label{sec:conclusions}

We have presented novel truthful and competitive mechanisms
for three important combinatorial problems: Vertex Covers, $k$-flows,
and $s$-$t$ cuts. The Vertex Cover mechanism was based on scaling the
submitted bids by multipliers derived as components of the 
dominant eigenvector of a suitable matrix.
Both the flow and cut mechanisms were based on pruning the input graph,
and then applying the Vertex Cover mechanism to the pruned version.
Besides the individual mechanisms, we believe that the methodology of
reducing input instances to Vertex Cover problems may be of interest
for future frugal mechanism design.

In general, the Vertex Cover mechanism does not run in polynomial
time, due to two obstacles. First, computing the matrix \MAT
requires computing the largest fractional clique size in the
neighborhood of each node $v$.
Subsequently, computing the solution with respect to scaled
costs requires finding a cheapest vertex cover. For the second
obstacle, it seems quite likely that monotone algorithms such as the
one in \cite{elkind:goldberg:goldberg:frugality} could be adapted to
our setting, and yield constant-factor approximations.
However, the difficulty of computing the entries of \MAT
seems more severe. In fact, we conjecture that no polynomial-time
truthful mechanism for Vertex Cover can be constant-competitive.
This result would be quite interesting, in that it would show that the
requirements of incentive-compatibility and computational tractability
together can lead to significantly worse guarantees than either
requirement alone. 
A positive resolution of this conjecture would thus
be akin to the types of hardness results demonstrated recently for
the Combinatorial Public Project Problem
\cite{papadimitriou:schapira:singer:hardness}.

While our methodology of designing composable pre-processing
algorithms will likely be useful for other problems as well, it does
not apply to all set systems. It is fairly easy to construct set
systems for which no such pruning algorithm is possible.
Even when pruning is possible in principle, it may come with a large
blowup in costs.

Thus, the following bigger question still stands: which classes of
set systems admit constant-competitive mechanisms?
The main obstacle is our inability to prove strong lower bounds on frugality
ratios. To date, all lower bounds (here, as well as in
\cite{elkind:sahai:steiglitz,BeyondVCG}) are based on pairwise
comparisons between agents, which can then be used to show that
certain agents, by virtue of losing, will cause large payments.
This technique was exactly the motivation for our Vertex Cover approach.
In order to move beyond Vertex Cover based mechanisms, it will be
necessary to explore lower bound techniques beyond the one used in
this paper.

In recent joint work with the authors of
\cite{chen:elkind:gravin:petrov}, we have shown that the 
factor \half in the lower bound of Lemma \ref{lem:vertex-lower} can be
removed, thus showing that \EVMECH is optimal.
The proof of this result will be presented in a joint full version of
the present paper with \cite{chen:elkind:gravin:petrov}.

\subsection*{Acknowledgment}
We would like to thank Edith Elkind, Uriel Feige, Nick Gravin, Anna
Karlin, Tami Tamir, and Mihalis Yannakakis for useful discussions and
pointers, and anonymous reviewers for useful feedback.  
D.K.~is supported in part by an NSF CAREER Award, an ONR Young
Investigator Award and an award from the Sloan Foundation.
C.M.~is supported in part by the McDonnell Foundation.



%
\bibliographystyle{plain}
\bibliography{names,conferences,publications,bibliography}

\end{document}